\documentclass[]{article}

\usepackage[a4paper, margin=3cm]{geometry}
\usepackage{tikz}
\usepackage{tikz-network}
\usetikzlibrary{positioning, trees}
\usepackage{algorithm, algpseudocode}
\usepackage{todonotes}
\usepackage{multirow}
\usepackage{thmtools}
\usepackage{hyperref}
\usepackage{amsthm,amssymb}
\usepackage{esvect}
\usepackage{booktabs}

\newtheorem{definition}{Definition}
\newtheorem{theorem}{Theorem}

\newtheorem{example}{Example}

\newtheorem{proposition}{Proposition}
\newtheorem{lemma}{Lemma}

\AtBeginDocument{%
  }

\newcommand{\tree}{\mathcal{T}}
\newcommand{\nodes}{\mathcal{N}}
\newcommand{\edges}{\mathcal{E}}
\newcommand{\items}{\mathcal{G}}

\newcommand{\children}{\mathcal{C}}
\newcommand{\parent}{\mathcal{P}}
\newcommand{\leaves}{\mathcal{L}}
\newcommand{\internal}{\mathcal{I}}
\newcommand{\ancestors}{Anc}

\newcommand{\fair}{\mathcal{F}}

\newcommand{\mallocations}{\mathrm{\Pi}}
\newcommand{\allocations}{\mathcal{A}}

\date{}

\begin{document}

\title{\bfseries Multilevel Fair Allocation with Matroid-Rank Preferences}

\author{
Maxime Lucet, Nawal Benabbou, Aur\'elie Beynier, and Nicolas Maudet 
\\
LIP6, CNRS, Sorbonne Université\\
\texttt{\{firstname.lastname\}@lip6.fr}
}

\maketitle
\begin{abstract}
    We introduce the concept of multilevel fair allocation of resources with tree-structured hierarchical relations among agents.
    While at each level it is possible to consider the problem locally as an allocation of an agent to its children, the multilevel allocation can be seen as a trace capturing the fact that the process is iterated until the leaves of the tree. In principle, each intermediary node may have its own local allocation mechanism.
    The main challenge is then to design algorithms which can retain good fairness and efficiency properties.
    In this paper, we propose two original algorithms under the assumption that leaves of the tree have matroid-rank utility functions and the utility of any internal node is the sum of the utilities of its children. The first one is a generic polynomial-time sequential algorithm that comes with theoretical guarantees in terms of efficiency and fairness. It operates in a top-down fashion -- as commonly observed in real-world applications -- and is compatible with various local algorithms. The second one extends the recently proposed General Yankee Swap to the multilevel setting. This extension comes with efficiency guarantees only, but we show that it preserves excellent fairness properties in practice.
    We also show that both algorithms can be extended to handle chores for some fairness criteria. 
\end{abstract}

\section{Introduction}

A multiagent resource allocation problem consists in determining a fair and/or efficient distribution of a set of items among a set of agents \cite{Brandt2016}. Most existing work assumes that items are allocated directly to the agents or groups who will use them. However, in many real-world settings, agents are embedded within hierarchical organizations: groups nested within larger groups, forming a multilevel hierarchical organization. Such a hierarchy can be naturally modeled as a directed tree, where each node is responsible for allocating the bundle of items it receives to its children.

Consider the allocation of offices to researchers within a university exhibiting a natural hierarchical structure: the university is composed of departments, each containing laboratories, which in turn comprise research groups made up of individual researchers. The allocation process typically proceeds in stages: offices are first distributed among departments, then among laboratories within each department, then by research groups within each laboratory, and finally among individual researchers. In such a setting, fairness should be enforced within each hierarchical level (e.g., among departments or among research groups within a laboratory), rather than across levels, since entities at different levels represent incomparable units. This example extends to many allocation problems across territories and hierarchical organizations. 

Applications can arise even in the absence of an explicitly defined hierarchy. For instance, when allocating jobs to unemployed individuals characterized by gender, age, and education, one may first focus on fairness across gender (e.g., putting more weight on women to address underrepresentation), then across education for women (to promote more qualified women) and across age for men by assigning equal weight to each age group (to balance the representation). Such hierarchical preferences can be naturally represented in the multilevel framework as a directed tree in which the nodes are weighted. Note that this last example is closely related to real-world affirmative action policies, such as the reservation system in India\footnote{\url{https://en.wikipedia.org/wiki/Reservation_in_India}}\footnote{\url{https://www.pib.gov.in/Pressreleaseshare.aspx?PRID=1564231&reg=3&lang=2}}, which aims to address historical inequalities. In this system, a fixed fraction of positions is reserved for individuals from disadvantaged groups (e.g., Scheduled Castes, Scheduled Tribes, and Other Backward Classes), with further refinements sometimes introduced to prioritize more marginalized subgroups. In addition, horizontal reservations (e.g., for women or persons with disabilities) apply across these categories. Although the resulting structure is not strictly hierarchical, it can be naturally represented within our framework as a weighted tree capturing successive layers of prioritization.

Different assumptions can be made regarding the utility of the nodes (agents) of the tree. For internal nodes, it could be (i) that agents only care about the bundle they receive and are oblivious of the subsequent allocation of those items; (ii) that agents have preferences over the properties of the allocation to their children; or (iii) that agents have preferences that depend on the actual utilities of their children. 
In this paper, we adopt assumption (iii) which reflects the idea that internal entities do not act as self-interested agents, but rather as representative bodies for their constituents. 
For example, a department’s objective is to maximize the aggregate satisfaction of the labs it contains, and similarly a laboratory aims to maximize the satisfaction of its research groups. In this paper, we model each internal node as maximizing the {\em utilitarian social welfare}, that is, its utility is defined as the sum of the utilities of its children. We acknowledge that, in practice, some intermediate entities (internal nodes) may display a degree of self-interest; for instance, a laboratory might wish to retain additional space to support its own development. Nevertheless, modeling these nodes as purely representative agents, without an intrinsic utility of their own, provides a natural first approach. Moreover, this assumption already raises notable algorithmic challenges when allocations are computed using a top-down approach: an internal agent must allocate items to satisfy its children, even though their utilities are not yet known, as these depend on subsequent allocations at lower levels, and so on. 

Finally, we adopt \textit{matroid-rank} utilities for the leaves, a well-established class of valuations that succinctly capture preferences for diversity over items. For example, when allocating offices to researchers, research groups can be modeled as the leaves of the hierarchy, each consisting of a set of researchers. The utility of a group is then defined as the size of a maximum matching in a bipartite graph: one side represents the researchers in the group, the other the offices allocated to the group, and an edge connects a researcher to an office if the researcher finds the office acceptable. This matching captures how well the allocated offices satisfy the preferences of the group’s members. Such utilities can be modeled by matroid-rank utilities. Other examples include students favoring curricula spanning multiple disciplines, or allocation problems with diversity constraints which may prohibit exceeding category quotas 
\cite{Benabbou_et_al_2019}. Unlike additive utilities, matroid-rank valuations have limited expressive power over singletons but gain richness through their submodular structure, which captures complementarities and diminishing returns. A remarkable property of matroid rank domains is that they allow the existence of allocations satisfying simultaneously efficiency and fairness, as observed in several papers \cite{Benabbou_et_al2021, Babaioff_et_al2021, Viswanathan_Zick2023b, Viswanathan_Zick2023a}. In this paper, we show that this desirable property can be extended to our multilevel setting.

We now present two in-depth examples illustrating how the multilevel structure, and matroid-rank preferences at the leaves can coexist in real-world applications. 

\subsection{COVID-19 vaccine allocation}

A real-world illustration of our model is the problem of fairly distributing vaccines during the COVID-19 pandemic. Ensuring territorial fairness mattered not only for moral reasons, since without coordination richer countries could secure more doses than needed at the expense of poorer ones, but also for epidemiological reasons, as global vaccination is key to preventing resurgence. Similar concerns arose at the national level: in France, for example, some members of the Parliament pointed out disparities in vaccine allocation across regions, with certain areas receiving significantly more doses than others\footnote{\url{https://questions.assemblee-nationale.fr/q15/15-36657QE.htm}}. Such challenges of fair territorial distribution can be naturally addressed within the framework we propose, which allows fairness considerations to be incorporated across multiple administrative levels. In particular, the allocation should be fair between countries, for instance proportionally to population, and recursively fair across successive territorial subdivisions within each country, ultimately down to the level of hospitals or vaccination centers.

Moreover, preferences of hospitals or vaccination centers can be naturally modeled using matroid-rank functions. For instance, consider a setting in which a fixed stock of vaccine doses must be distributed among hospitals or vaccination centers worldwide. Each hospital is subject to a storage capacity constraint, limiting the total number of doses it can safely hold. In addition, hospitals may approve or disapprove of specific vaccines; during the COVID-19 pandemic, for example, some countries declined or restricted the use of vaccine doses from certain foreign suppliers. Consequently, each hospital seeks to receive as many doses as possible among the vaccines it approves, while respecting its capacity constraints.

\subsection{Course allocation}

Another real-world illustration and common motivation for fair allocation problems with matroid-rank valuations is the assignment of courses to students \cite{Benabbou_et_al2021, Viswanathan_Zick2023a}. In such settings, students may approve or disapprove of courses, but their preferences are more structured: each student has a capacity constraint (e.g., a maximum number of courses per semester), some courses may be incompatible (e.g., due to overlapping schedules), and students may have additional structured requirements (e.g., wishing to take at most a certain number of mathematics courses). These preferences can naturally be modeled using matroid-rank valuations. Furthermore, \cite{Bissias_et_al2025} show that many real-world students exhibit preferences that are indeed matroid-rank valuations. Finally, \cite{Bissias_et_al2025} also shows that even when preferences are not exactly matroid-rank, they often retain much of the underlying matroid-rank structure, allowing algorithms designed for matroid-rank valuations to maintain strong empirical performance. 

In this context, the university may also wish to ensure fairness or implement prioritization among students, leading to an implicit hierarchical structure. For instance, it may prioritize students closer to graduation over first-year students, enforce strict gender parity between male and female students, and incorporate a merit-based preference that increases the chances for higher-performing students to obtain desired courses. Such requirements can be captured within the framework we introduce in this paper.

\subsection{Related work}
Recently, \emph{group fairness} has become an active topic in fair division. For example, \cite{Aziz_Rey_2019,Conitzer_et_al2019} studied group fairness notions requiring fairness to be achieved among \emph{any} group, and introduced different relaxations of group envy-freeness. When groups are \emph{pre-defined}, several notions and algorithms have also been proposed \cite{Benabbou_et_al_2019, Chakraborty_et_al_2020,Kyropoulou_2019, Aleksandrov_Walsh_2018, GrossHumbert_et_al2023}. 

We focus on leaves with matroid-rank valuations. Several papers study classical (monolevel) allocation problems where agents have such utilities \cite{Benabbou_et_al2021, Babaioff_et_al2021, Viswanathan_Zick2023a}. In particular, our second algorithm extends the General Yankee Swap \cite{Viswanathan_Zick2023b} to our multilevel setting.

Also closely related in spirit to our contribution is the work of \cite{scarlett_et_al2023,Bu_et_al2024,Aggarwal_et_al2024} which seeks to reconcile the perspectives of groups and individual agents. 
More precisely, \cite{scarlett_et_al2023} study individual and group envy-freeness in settings where groups' valuations depend on the valuations of their constituent agents, and they propose algorithms that provide guarantees at both levels. By contrast,
\cite{Bu_et_al2024}  argue that agents at different levels may have distinct preferences, notably that groups' preferences do not necessarily aggregate the valuations of their constituent agents (while we do assume that nodes aggregate the utility of their children in this paper). Lastly, \cite{Aggarwal_et_al2024} studies a bilevel allocation problem, motivated by the way some food charities operate. In contrast with our paper, their paper focuses on auction mechanisms. All of these papers remain limited to two-level settings. 

Furthermore, our setting resembles the recent multilevel apportionment model from \cite{Schmidt_et_al2025}. Indeed, their hierarchical model is similar to ours, representing the hierarchy through a tree. However, the scope of the paper is specifically on apportionment methods, a specific subfield of resource allocation.

Our model also shares similarities with the trend of work on local fairness \cite{Abebe_et_al_2017,Beynier_et_al2019, BREDERECK_2022}, as the hierarchical structure induces some sort of visibility constraint (e.g., only labs affiliated to the same department can compare their situations). 
 
\subsection{Contributions}
In this paper, we formally introduce multilevel fair allocation problems and propose two algorithms to address them. We assume that the leaves of the tree (e.g., research groups or unemployed individuals) have matroid rank utility functions, and that each internal node aims to maximize the utilitarian social welfare (i.e., its utility is equal to the sum of its children’s utilities). The first algorithm introduced in Section~\ref{section: sma} is a top-down, polynomial-time method satisfying several fairness and efficiency properties. However, as discussed in Section~\ref{section: simulations}, it can be slow in practice. Hence, we propose a faster algorithm in Section~\ref{section: mgys} extending the General Yankee Swap \cite{Viswanathan_Zick2023b}, which ensures efficiency and exhibits strong empirical fairness. Moreover, we show in Appendix~\ref{appendix: chores} that both algorithms can be extended to handle chores for some fairness criteria, following the techniques of \cite{Viswanathan_Zick2023b}.

\section{Model} \label{section: model}

In this paper, we consider allocation problems where a set of $m$ items, denoted by $\items = \{g_1, \ldots, g_m\}$, must be distributed among agents organized in a hierarchical structure. In the main body of the paper, we consider goods to be allocated among agents, meaning that they have non-negative marginal utilities for these items, while in the appendix, we study the case where the items are chores.

\vspace{0.1cm}

\noindent \textbf{Hierarchical structure.} We consider a multilevel allocation problem represented by an arborescence (i.e., a directed rooted tree in which all edges point away from the root), denoted by $\tree = (\nodes, \edges)$. Here, $\nodes = \{1,\ldots,n\}$ is the set of nodes representing agents, and $\edges$ is the set of directed edges that represent the hierarchical relationships between them. We assume that the nodes are indexed according to a topological ordering of $\tree$, so that node $1$ is the root. Moreover, each node $i \in \nodes$ is associated with a weight $w_i > 0$ representing the entitlement of agent $i$, as required by some fairness criteria (the unweighted case can be treated by setting equal weights between nodes with the same parent). 

Let $\children(i)$ denote the set of children of node $i$, defined as $\children(i) = \{j \in \nodes : (i,j) \in \edges\}$. Let $\parent(i)$ denote the parent of $i$, i.e., the unique node such that $(\parent(i), i) \in \edges$. We denote by $\ancestors(i)$ the set of ancestors of $i$, that is, the nodes on the unique path from the root $1$ to $i$, excluding $i$. For any node $i$, let $\tree_i = (\nodes_i, \edges_i)$ denote the subtree of $\tree$ rooted at $i$, consisting of all nodes and edges belonging to paths that start at $i$. Let $\leaves(i)$ be the set of leaves of $\tree_i$, formally defined as $\leaves(i) = \{j \in \nodes_i : \children(j) = \emptyset\}$. Let $\internal(i)$ denote the set of internal nodes of $\tree_i$, defined by $\internal(i) = \nodes_i \setminus \leaves(i)$. In the particular case of the root, we simply write $\internal := \internal(1)$ and $\leaves := \leaves(1)$ to denote, respectively, the sets of internal nodes and leaves of the entire tree.

\medskip

\noindent \textbf{Allocations.} Now we introduce relevant allocation types:

\begin{definition}[multilevel allocation]
$\pi : \nodes \rightarrow 2^\items$ is a \emph{multilevel allocation} if it satisfies the following properties:
\begin{enumerate}
    \item $\pi(1) = \items$,
    \item $\pi(i) \supseteq \bigcup_{j \in \children(i)} \pi(j)$ for all $i \in \internal$,
    \item $\pi(i) \bigcap \pi(j)  = \emptyset$ for all $i, j \in \nodes$ such that $\parent(i) =  \parent(j)$, 
\end{enumerate}
where $\pi(i)$ denotes the bundle of items received by node $i \in \nodes$ under $\pi$. 
\end{definition}

\noindent We require that the root owns all items (i.e. $\pi(1) = \items$) only to ensure that none are discarded a priori. Note that we do not impose $\pi(i) = \cup_{j \in \children(i)} \pi(j)$ for all $i  \in  \internal$ allowing nodes not to allocate all their items, particularly to satisfy their fairness criterion. Instead, we only impose  $\pi(i) \supseteq \cup_{j \in \children(i)} \pi(j)$ to make the model as general as possible; for example, it captures allocation problems with charity \cite{Caragiannis_et_al2019b}, and it aligns with most studies on matroid rank valuations \cite{Benabbou_et_al2021, Viswanathan_Zick2023b}, which typically restrict attention to non-redundant allocations (i.e., only items with positive marginal utility are assigned). Finally, we require that each internal node allocate each of its items to at most one child. Hereafter, the set of all multilevel allocations is denoted by $\mallocations$.

\begin{definition}[restricted multilevel allocation] For a multilevel allocation $\pi   \in   \mallocations$ and a set of nodes $N   \subseteq  \nodes$, the \emph{restriction of $\pi $ to $N$} is denoted by $\pi |_N$ and is defined by $   \pi |_N =  (\pi(i))_{i \in N}$.
\end{definition}

\begin{definition}[local allocation] Given a set of nodes $N \subseteq  \nodes$ and a set of items $S  \subseteq  \items$, $A  :  N  \rightarrow  2^S$ is a \emph{local allocation} if \ $\bigcup_{i \in N} A(i)  \subseteq  S$ and $A(i) \bigcap A(j)  =  \emptyset$ for all $i, j  \in  N$.
\end{definition}

\noindent A local allocation is therefore simply a standard allocation (in the sense that the tree structure is ignored), where each item in $S$ is assigned to at most one agent in 
$N$, and some items may remain unallocated. For any $N  \subseteq  \nodes$ and $S  \subseteq  \items$, let $\mathcal{A}^S_N$ denote the set of corresponding local allocations. Note that, for any multilevel allocation $\pi \in \mallocations$ and any node $i \in \internal$, the restricted allocation $\pi |_{\children(i)}$ is a local allocation belonging to $\allocations^{\pi(i)}_{\children(i)}$. Indeed, node $i$ can only distribute to its children the goods it has received, namely $\pi(i)$, and each good can be assigned to at most one child, since in a multilevel allocation no good can be given to more than one agent. The notion of a local allocation will be particularly useful for comparing the different ways a node can distribute goods to its children, especially in terms of fairness and efficiency.

\medskip

\noindent \textbf{Utility model.} Let $v_i : \mallocations \rightarrow \mathbb{R}_{\geq 0}$ be the utility function of node $i  \in  \nodes$, and let $v =(v_i)_{i \in \nodes}$. Note that the utility function is defined over the set of multilevel allocations (rather than over bundles of items), since a node’s utility may depend on how goods are allocated to its children, or to any other node in the tree. In this paper, for any internal node $i \in \internal$ and any multilevel allocation $\pi \in \mallocations$, $v_i(\pi)$ quantifies some concept of overall welfare derived from the children of $i$ from $\pi $.  We focus here on the utilitarian social welfare, i.e. $$
v_i(\pi) := \sum_{j \in \children(i)} v_{j}(\pi).
$$ 
Hence, we assume that internal nodes have additive utilities over their children. This assumption is realistic in many applications where internal nodes represent entities whose utility derives solely from the satisfaction of their members. Note that, by linearity of summation, $v_i(\pi)$ can be rewritten as $$v_i(\pi) = \sum_{x \in \leaves(i)} v_x(\pi)$$ where the sum ranges over the leaves of the subtree $\tree_i$. In certain applications, we may consider more complex models that incorporate a personal component in addition to the dependency on the children. As already discussed, we leave this extension for future work.

In contrast, the utility of each leaf $x \in \leaves$ only depends on the bundle of items it receives (since it has no children). Hence there exists a function $u_x : 2^\items \rightarrow \mathbb{R}_{\geq 0}$ such that $v_x(\pi) = u_x(\pi(x))$ for any multilevel allocation $\pi \in \mallocations$, and we assume that $u_x(\emptyset) = 0$. We also assume that $u_x$ belongs to the class of matroid rank functions, also known as binary submodular utilities, defined as follows:

\begin{definition}[matroid rank function] \label{def: mrf}
    A set function $u: 2^\items \rightarrow \mathbb{R}_{\geq 0}$ is a \emph{Matroid Rank Function} (MRF) if it satisfies:
    \begin{enumerate}
        \item Monotonicity: $u(S) \leq u(T)$  for any $S \subseteq T \subseteq \items$, 
        \item Submodularity: $\Delta^u(S, g) \geq \Delta^u(T, g)$ for any $S  \subseteq T \subseteq \items$ and any $g \in \items \setminus T$, 
        \item Binary marginal gain: $\Delta^u(S, g)  \in \{0, 1\}$ for any $S \subseteq \items$ and any $g \in \items \setminus S$,
    \end{enumerate}
    where $\Delta^u(S, g) = u(S \cup \{g\}) - u(S)$ is the marginal gain of $g$ given $S$, for any $S \subseteq \items$ and any $g \in \items \setminus S$.
\end{definition}

\begin{example} \label{example: model}
    A university has to allocate 6 offices $\items=\{a,b,c,d,e,f\}$ to two departments (Humanities and Computer Science), which in turn will distribute them among their respective labs (see Fig.~\ref{fig:exampletree}). 

    
    \begin{figure}[hbtp]
    \centering
    \begin{tikzpicture}[
            level distance=12mm,
            every node/.style={rectangle, rounded corners=2pt, draw, inner sep=2pt, font=\footnotesize},
            level 1/.style={sibling distance=43mm},
            level 2/.style={sibling distance=21mm},
            level 3/.style={sibling distance=6mm},
            edge from parent/.style={draw, -stealth}
        ]
        
        \node (1) {1: University}
            child { node (2) {2: DeptH}
                child { node (4) {4: LabH1} }
                child { node (5) {5: LabH2} }
            }
            child { node (3) {3: DeptCS}
                child { node (6) {6: LabCS1} }
                child { node (7) {7: LabCS2} }
            };
        
    \end{tikzpicture}
    \caption{The hierarchical structure of a university.}
    \label{fig:exampletree}
    \end{figure}


    \noindent The labs are the leaves. Their matroid-rank utilities are as follows: LabH1 gets utility~1 upon receiving any office; LabH2 behaves similarly, except for ~$a$; LabCS1 and LabCS2 have binary additive valuations and approve all items. 
    \noindent Suppose that DeptH and DeptCS receive $\{a,b\}$ and $\{c,d,e,f\}$ respectively, and then assign their items so that LabH1, LabH2, and LabCS1 receive  $\{a\}$, $\{b\}$, and $\{c,d,e,f\}$, respectively, leaving LabCS2 with an empty bundle. This multilevel allocation, denoted by $\pi$, verifies: 

    \begin{table}[h!]
        \centering
        \begin{tabular}{c|cccccccc}
            \hline 
            Node $i$ & 1 & 2 & 3 & 4 & 5 & 6 & 7 \\
            \hline 
            $\pi(i)$ & $\{a, b, c, d, e, f\}$ & $\{a, b\}$ & $\{c, d, e, f\}$ & $\{a\}$ & $\{b\}$ & $\{cdef\}$ & $\emptyset$ \\ \hline 
            $v_i(\pi)$ & 6 & 2 & 4 & 1 & 1 & 4 & 0 \\
            \hline 
        \end{tabular}
    \end{table}

    One might ask whether the allocation performed by the university is efficient and fair, in a certain sense, with respect to the departments, by considering the restricted allocation $\pi|_{\children(1)}$ (which is the local allocation in $\allocations^\items_{\children(1)}$ defined by $\pi|_{\children(1)}(2)= \{a,b\}$ and $\pi|_{\children(3)}(3)= \{c,d,e,f\}$).

\end{example}

We aim to compute a multilevel allocation that simultaneously ensures high welfare and satisfies a fairness criterion with respect to the children of each (internal) node.

\medskip

\noindent \textbf{Efficiency criteria.}  We now introduce the efficiency notions considered in this paper.

\begin{definition}[non-redundancy w.r.t. $v$] \label{definition: non-redundancy} 
Multilevel allocation $\pi \in \mallocations$ is \emph{non-redundant} if 
$v_i(\pi) - v_i(\pi_{-g}) > 0$ for all $i \in \nodes$ and all $g \in \pi(i)$, where $\pi_{-g}$ denotes the multilevel allocation obtained from $\pi$ by removing $g$ from all bundles $\pi(i)$, $i \in \nodes$, that include $g$.
\end{definition}

\begin{definition}[multilevel utilitarian optimality w.r.t. $v$] \label{def: multilevel usw} Given any multilevel allocation $\pi  \in  \mallocations$ and any internal node $i \in \internal$, $\pi \vert_{\children(i)}$ is said to be \emph{utilitarian optimal w.r.t. $v$} if and only if for all $\pi' \in \mallocations$ such that $\pi'(i) = \pi(i)$, we have:
$$
\sum_{j \in \children(i)} v_{j}(\pi) \ge \sum_{j \in \children(i)} v_{j}(\pi')
$$

\noindent Then $\pi \in \mallocations$ is said to be \emph{multilevel utilitarian optimal w.r.t. $v$} if $\pi \vert_{\children(i)}$ is utilitarian optimal w.r.t. $v$ for all $i\in \internal$.
\end{definition}

Note that a multilevel utilitarian optimal allocation can be computed in polynomial time, thanks to the additivity of $v_i$, for all $i \in \internal$. Indeed, it suffices to find a local allocation $A  \in \allocations^\items_{\leaves}$ that maximizes $\sum_{x \in \leaves} u_x(A(x))$, and then define $\pi(x)=A(x)$ for all leaves $x \in \leaves$ and $\pi(i) = \cup_{x \in \leaves(i)} A(x)$ for each internal node $i\in \internal$. Since each $u_x$ is a matroid rank function, $A$ can be computed in polynomial time: this problem was already addressed in fair allocation by previous papers \cite{Benabbou_et_al2021, Babaioff_et_al2021}, which show that it can be solved in polynomial time by reducing it to a \textit{matroid intersection problem} \cite{Chakrabarty_et_al2019}. In the proof of Theorem~\ref{theorem: sma terminates}, we show that it can be solved even faster using the \textit{matroid partition problem}, whose best available algorithm has lower complexity than the matroid intersection problem \cite{Terao2025}. Nevertheless, we emphasize that a  multilevel utilitarian allocation is not necessarily fair among the children of each node, as illustrated later in Example \ref{example2} (after the fairness notions are introduced).

\medskip

\noindent {\bf Fairness criteria.} We consider fairness criteria comparing allocations through utility vectors. In this paper, we consider a set $\fair$ of four well-known fairness criteria that enable the comparison of any two utility vectors, thereby inducing a total order $\succeq$ over the possible local allocations: Lorenz-dominance, maximum weighted Nash social welfare, weighted leximin, and maximum weighted p-means welfare. We define a node-specific fairness criterion $\Psi_i$ for each internal node $i \in \internal$. In our setting, this amounts to comparing utility vectors of local allocations at node $i$.  Specifically, given a multilevel allocation $\pi \in \mallocations$, the utility vector of $\pi \vert_{\children(i)}$ is defined by $\vv{v_i(\pi \vert_{\children(i)})} = (v_{i_1}(\pi), v_{i_2}(\pi), \ldots, v_{i_p}(\pi))$ where $i_1, \ldots, i_p$ are the children of $i$, ordered by their index. Given two multilevel allocations $\pi , \pi' \in \mallocations$, we say that $\pi \vert_{\children(i)}$ is fairer w.r.t. $\Psi_i$ and $v$ than $\pi ' \vert_{\children(i)}$ if $$\vv{v_i(\pi \vert_{\children(i)})} \succeq_{\Psi_i} \vv{v_i(\pi' \vert_{\children(i)})}$$ also denoted by $\pi \vert_{\children(i)} \succeq_{\Psi_i}^v \pi' \vert_{\children(i)}$ for brevity. Let $\Psi=(\Psi_i)_{i \in \internal}$. 
Now we define the associated multilevel fairness notion. 

\begin{definition}[multilevel $\Psi$-maximizing w.r.t. $v$]  \label{def: multilevel psi} Given any multilevel allocation $\pi \in \mallocations$ and any internal node $i \in \internal$, $\pi \vert_{\children(i)}$ is said to be \emph{$\Psi_i$-maximizing w.r.t. $v$} if for all $\pi' \in \mallocations$ such that $\pi'(i) = \pi(i)$, we have:
$$
\pi \vert_{\children(i)} \succeq_{\Psi_i}^v \pi' \vert_{\children(i)}
$$

\noindent Then $\pi \in \mallocations$ is said to be \emph{multilevel $\Psi$-maximizing w.r.t. $v$} if $\pi \vert_{\children(i)}$ is $\Psi_i$-maximizing w.r.t. $v$ for all $i \in \internal$.
\end{definition}

\noindent Below are the definitions of the four fairness notions considered in the paper, adapted to our multilevel setting.

\begin{definition}[Lorenz-dominance w.r.t. $v$] Given a multilevel allocation $\pi \in \mallocations$ and an internal node $i \in \internal$ whose children are $\{i_1,\ldots, i_p\}$, $\pi \vert_{\children(i)}$ is \emph{Lorenz-dominating w.r.t. $v$} if for all $\pi' \in \mallocations$ such that $\pi'(i) = \pi(i)$, we have:
$$ \forall k \in \{1,\ldots, p\}, 
\sum_{t=1}^k \vv{s(\pi|_{\children(i)})}_t \geq \sum_{t=1}^k \vv{s(\pi'|_{\children(i)})}_t
$$
where $\vv{s(\cdot)}$ is the vector $\vv{v_i(\cdot)}$ sorted in increasing order, and $\vv{s(\cdot)}_t$ is the $t^{\text{th}}$ component of $\vv{s(\cdot)}$ for any $t \in \{1, \ldots, p\}$.
\end{definition}

\begin{definition}[maximum weighted Nash social welfare w.r.t. $v$] Given a multilevel allocation $\pi \in \mallocations$ and an internal node $i \in \internal$, $\pi \vert_{\children(i)}$ \emph{maximizes weighted Nash social welfare w.r.t. $v$} if it minimizes the number of children with zero utility, and subject to that, satisfies for all $\pi' \in \mallocations$ such that $\pi'(i) = \pi(i)$:  
$$
\prod_{j \in \children(i)} v_{j}(\pi)^{w_{j}} \geq \prod_{j \in \children(i)} v_{j}(\pi')^{w_{j}}
$$
\end{definition}

\begin{definition}[weighted leximin w.r.t. $v$] Given a multilevel allocation $\pi \in \mallocations$ and an internal node $i \in \internal$ whose children are $\{i_1,\ldots, i_p\}$, $\pi \vert_{\children(i)}$ is \emph{weighted leximin w.r.t. $v$} if there exists no $\pi' \in \mallocations$ with $\pi'(i) = \pi(i)$ such that:
$$
\vv{e(\pi' \vert_{\children(i)})} \succ_{lex} \vv{e(\pi \vert_{\children(i)})}
$$
where $\vv{e(\cdot)}$ is the vector $(\frac{\vv{v_i(\cdot)}_1}{w_{i_1}}, \ldots, \frac{\vv{v_i(\cdot)}_p}{w_{i_p}})$ sorted in increasing order, and $\succ_{lex}$ is the lexicographic dominance\footnote{Given two vectors $x, y \in \mathbb{R}^c$ for some positive integer $c$, $x$ \emph{lexicographically dominates} $y$ if there exists $k \in \{1, \ldots, c\}$ such that $x_j = y_j$ for all $j \in \{1, \ldots, k-1\}$ and $x_k > y_k$.}. 
\end{definition}

\begin{definition}[maximum weighted $p$-means welfare w.r.t. $v$] Given a multilevel allocation $\pi \in \mallocations$ and an internal node $i \in \internal$, $\pi \vert_{\children(i)}$ \emph{maximizes weighted $p$-means welfare w.r.t. $v$} for some $p \leq 1$ if it minimizes the number of children with zero utility, and subject to that satisfies for all $\pi' \in \mallocations$ such that $\pi'(i) = \pi(i)$:
$$
\Big (\sum_{j \in \children(i)} w_{j} v_{j}(\pi)^p \Big )^{\frac{1}{p}} \geq \Big (\sum_{j  \in \children(i)}  w_{j}  v_{j}(\pi')^p \Big )^{\frac{1}{p}}$$
\end{definition}

Interestingly, these four notions are compatible with utilitarian social welfare optimization when the valuation functions are known to be MRF \cite{Benabbou_et_al2021, Viswanathan_Zick2023b}. A key step in our first approach is to define suitable valuation functions for internal nodes that satisfy this property throughout the execution of the algorithm; these will be denoted by $\hat{v}_i$, $i\in \nodes$, in the following section. 

\medskip

We now illustrate the potential lack of fairness among the children of a node under multilevel utilitarian allocations, highlighting the need to design algorithms that account for both welfare optimization and fairness.

\begin{example} \label{example2}
   We use the same instance than Example~\ref{example: model}. Let us compare the following multilevel allocations (we only report the bundles of the labs, as the bundles of the departments can then be obtained by taking the union of the bundles of their constituent labs): 

\vspace{0.1cm}

\noindent
\begin{center}
\scalebox{0.9}{
$
\begin{array}{c|ccccc} \hline
\text{Alloc} & LabH1& LabH2 & LabCS1 & LabCS2 & \text{utilities} \\  \hline
\pi & \{a\} & \{b\} & \{c,d,e,f\} & \emptyset & \langle 1,1,4,0\rangle \\
\pi' & \{b\} & \{a\} & \{c,d,e,f\} & \emptyset & \langle 1,0,4,0\rangle \\
\pi'' & \{a\} &\{b\} & \{c,d\} & \{e,f\} & \langle 1,1,2,2 \rangle \\ \hline
\end{array}
$
}
\end{center}
\vspace{0.1cm}

\noindent First, notice that $\pi'$ is not utilitarian optimal at $DeptH$, as swapping $a$ and $b$ between $LabH1$ and $LabH2$ would yield utility 2 instead of 1. Then, it can be shown that 
$\pi$ is multilevel utilitarian optimal but not fair at node DeptCS considering Lorenz-dominance w.r.t. $v$ as it is dominated by $\pi''$. Finally, $\pi''$ is both multilevel utilitarian optimal and multilevel Lorenz-dominating w.r.t. $v$.

\end{example}

\section{Sequential Multilevel Algorithm} \label{section: sma}

We now propose a sequential multilevel algorithm that follows a top-down approach, drawing inspiration from real-life hierarchical systems (e.g. in vaccine rollout systems, doses are centrally purchased and then gradually distributed down the administrative hierarchy, from the state to regions, from regions to departments, and from departments to municipalities). Starting at the root, which is given the set $\items$, we compute a local allocation to the children of the root. Then, for each of its children, we compute a local allocation of the bundle it received from the root to its own children, and so on until the leaves. We will show that such an approach can construct a multilevel $\Psi$-maximizing and multilevel utilitarian optimal allocation w.r.t. $v$. 

\medskip

\noindent \textbf{Estimated utility.} In order to achieve efficiency/fairness at every node, each node should take into account the utilities of its children when assigning them bundles of items. However, in our setting, each internal agent $i \in \internal$ has a utility function $v_i$ which is defined from the set of possible multilevel allocations, not from the set of all bundles. Indeed, the utility of each child not only depends on the bundle it receives, but also on how it will be allocated to its own children (which also depends on the fairness notions considered). As a result, within a top-down procedure, the utility of an internal node is indeterminate until the leaves are reached, which underscores a key difficulty inherent in the top-down approach. To remedy this issue, our algorithm associates to each node $i \in \nodes$ a utility function $\hat{v}_i : 2^\items \rightarrow \mathbb{R}_{\geq 0}$ estimating the utility $v_i(\pi^*)$ derived by $i$ from the multilevel allocation $\pi^* \in \mallocations$ finally returned by the algorithm. This estimated utility function is used as a proxy of $v_i$. For any $S \subseteq \items$, the value $\hat{v}_i(S)$ is defined using the leaves of $\tree_i$ as follows:
$$
\hat{v}_i(S) :=\left\{
  \begin{array}{ll}
    \max_{A \in \allocations^S_{\leaves(i)}} \sum_{x \in \leaves(i)} u_x(A(x)), & \textit{if } i \in \internal \\
    u_i(S), & \textit{otherwise}.
  \end{array}
\right.
$$
We show in the proof of Theorem \ref{theorem: sma terminates} that $\hat{v}_i(S)$ can indeed be computed efficiently. Now to ensure that $\pi^*$ is multilevel utilitarian optimal and multilevel $\Psi$-maximizing w.r.t. $v$, our algorithm constructs a multilevel allocation that is multilevel utilitarian optimal and multilevel $\Psi$-maximizing w.r.t. $\hat{v}$ by sequentially identifying, at each $i \in \internal$, a local allocation that is utilitarian optimal and $\Psi_i$-maximizing w.r.t. $\hat{v}$. We give the definitions of these concepts based on $\hat{v}$ in the following. 

\begin{definition}[multilevel utilitarian optimality w.r.t. $\hat{v}$] \label{def: multilevel usw v hat} Given a bundle of items $S \subseteq \items$ and any internal node $i\in\internal$, local allocation $A \in \allocations^S_{\children(i)}$ is \emph{utilitarian optimal w.r.t. $\hat{v}$} if, for all $A' \in \allocations^S_{\children(i)}$, we have:
$$
\sum_{j \in \children(i)} \hat{v}_j(A(j)) \geq \sum_{j \in \children(i)} \hat{v}_j(A'(j))
$$

\noindent Multilevel allocation $\pi \in \mallocations$ is \emph{multilevel utilitarian optimal w.r.t. $\hat{v}$} if, for all internal nodes $i \in \internal$, 
$\pi \vert_{\children(i)}$ is utilitarian optimal w.r.t. $\hat{v}$ for $S = \pi(i)$.
\end{definition}

Given a set of items $S \subseteq \items$ and a set of nodes $N \subseteq \nodes$, we define the estimated utility vector of a local allocation $A \in \allocations^S_N$ as $\vv{\hat{v}_i(A)} = (\hat{v}_{i_1}(A(i_1)), \hat{v}_{i_2}(A(i_2)), \ldots, \hat{v}_{i_p}(A(i_p)))$ where $i_1, \ldots, i_p$ are the nodes of $N$ ordered by their index. At every internal node $i \in \internal$ with bundle $S \subseteq \items$, we use the notation $\succeq_{\Psi_i}^{\hat{v}}$ to compare two local allocations $A, B \in \allocations^S_{\children(i)}$ based on the utility vectors $\vv{\hat{v}_i(A)}$ and $\vv{\hat{v}_i(B)}$. This leads to the following multilevel fairness notions.

\begin{definition}[multilevel $\Psi$-maximizing w.r.t. $\hat{v}$] \label{def: multilevel psi v hat} Given a bundle of items $S \subseteq \items$ and an internal node $i \in \internal$, local allocation $A \in \allocations^S_{\children(i)}$ is \emph{$\Psi_i$-maximizing w.r.t. $\hat{v}$} if, for all $A' \in \allocations^S_{\children(i)}$, we have:
$$
A \succeq_{\Psi_i}^{\hat{v}} A'
$$
Then $\pi \in \mallocations$ is said to be {\em multilevel  $\Psi$-maximizing w.r.t. $\hat{v}$} if, for all $i \in \internal$, $\pi \vert_{\children(i)}$ is $\Psi_i$-maximizing w.r.t. $\hat{v}$  for $S = \pi(i)$. 
\end{definition}

It is worth noting the distinction between Def.~\ref{def: multilevel usw}, \ref{def: multilevel psi} and Def.~\ref{def: multilevel usw v hat}, \ref{def: multilevel psi v hat}: the latter are formulated in terms of $\hat{v}$ rather than $v$, and are defined with respect to local allocations $A \in \allocations^S_{\children(i)}$ as opposed to multilevel allocations. Appendix~\ref{appendix: fairness v-hat} gives the definitions corresponding to Lorenz-dominance, weighted leximin, maximum weighted Nash welfare, and weighted $p$-means welfare w.r.t. $\hat{v}$. 

\medskip

\noindent \textbf{Algorithm.} We can now formalise the functioning of the algorithm. Starting at node $i=1$, which is given the set $\items$, we compute a local allocation $A^i \in \allocations^{\items}_{\children(i)}$ which is utilitarian-optimal and $\Psi_{i}$-maximizing w.r.t. $\hat{v}$. Each child of the root $j \in \children(i)$ thus receives the bundle of items $A^i(j)$ and will in turn compute a local allocation in $\allocations^{A^i(j)}_{\children(j)}$ to its own children which is  utilitarian optimal and $\Psi_j$-maximizing w.r.t. $\hat{v}$. The process repeats until reaching the leaves. This Sequential Multilevel Algorithm (\textbf{SMA}) constructs a multilevel allocation from the initial call \textbf{SMA}$(1, \pi, v, \Psi)$ where $\pi $ is such that $\pi(1) = \items$ and $\pi(i) = \emptyset$ for all $i \in \nodes \backslash \{1\}$. The pseudocode can be found in Algorithm~\ref{algo:sequential_algo}. 

\begin{algorithm}[] \small
\caption{\small Sequential Multilevel Algorithm (SMA)}
\label{algo:sequential_algo}
\begin{algorithmic}[1]
    \State \textbf{Input:} $i$ -- a node in $\nodes$; $\pi $ -- a multilevel allocation; $v$ -- valuations of all nodes; $\Psi$ -- fairness criteria of internal nodes
    \State \textbf{Output:} $\pi $ -- a multilevel allocation
    \If{$\children \neq \emptyset$}
        \State Compute a local allocation $A^i \in \allocations^S_{\children(i)}$ that is utilitarian-optimal w.r.t. $\hat{v}$ and $\Psi_i$-maximizing w.r.t. $\hat{v}$ for $S = \pi(i)$.
        \State Set $\pi(j)=A^i(j)$ for all $j \in \children(i)$.
        \State Call \textbf{SMA}($j, \pi, v, \Psi$) for all $j \in \children(i)$.
    \EndIf
\end{algorithmic}
\end{algorithm}

We now illustrate this process with a simple execution of {\bf SMA}.

\begin{example} \label{example: sma}
Consider an instance with $m = 5$ items, i.e. $\items = \{g_1, \ldots, g_5\}$, and $n = 7$ agents arranged in a balanced binary tree: $\children(1) = \{2, 3\}$, $\children(2) = \{4, 5\}$, and $\children(3) = \{6, 7\}$ (as in Fig.~\ref{fig: initial pi}). Agents have (arbitrary) weights $w_2 = 5$, $w_3 = 2$, and $w_i = 1$ for $i\in\{4,5,6,7\}$. Let $\Psi_1$ be the maximum weighted Nash welfare and $\Psi_2$, $\Psi_3$ be Lorenz-dominance. Leaf utilities are binary additive: agents $4$, $5$, and $6$ only value items $\{g_1, g_2, g_3\}$ while agent $7$ only values $\{g_3, g_4, g_5\}$. The multilevel allocation is initialized as follows: $\pi(1)= \{g_1, \ldots, g_5\}$ and $\pi(i)= \emptyset$ for all $i\in \nodes \backslash \{1\}$ (see Figure~\ref{fig: initial pi}).

\begin{figure}[H]
    \centering
    \begin{tikzpicture}[scale=0.85,
            level distance=12mm,
            every node/.style={rectangle, rounded corners=2pt, draw, inner sep=5pt},
            level 1/.style={sibling distance=43mm},
            level 2/.style={sibling distance=21mm},
            level 3/.style={sibling distance=6mm},
            edge from parent/.style={draw, -stealth}
        ]
        
        \node (1) {1}
            child { node (2) {2}
                child { node (4) {4} }
                child { node (5) {5} }
            }
            child { node (3) {3}
                child { node (6) {6} }
                child { node (7) {7} }
            };
    
    \node[draw=none, text=black, above=0.1cm of 1] 
        {$\pi(1) = \{g_1, g_2, g_3, g_4, g_5\}$};
    
    \node[draw=none, text=black, left=0.1cm of 2] 
        {$\pi(2) = \emptyset$};
    
    \node[draw=none, text=black, right=0.1cm of 3] 
        {$\pi(3) = \emptyset$};

    \node[draw=none, text=black, below=0.1cm of 4] 
        {$\pi(4) = \emptyset$};
    
    \node[draw=none, text=black, below=0.1cm of 5] 
        {$\pi(5) = \emptyset$};

    \node[draw=none, text=black, below=0.1cm of 6] 
        {$\pi(6) = \emptyset$};

    \node[draw=none, text=black, below=0.1cm of 7] 
        {$\pi(7) = \emptyset$};
    \end{tikzpicture}
    \caption{Initial multilevel allocation $\pi$.}
    \label{fig: initial pi}
\end{figure}

The algorithm starts with the call {\bf SMA}$(1,\pi, v, \Psi)$, with  $\pi(1) = \items$ and $\pi(i) = \emptyset$ for all $i \neq 1$. It computes a utilitarian-optimal local allocation $A^{1} \in \allocations^\items_{\children(1)}$ maximizing the weighted Nash welfare w.r.t. $\hat{v}$, e.g., $A^{1}(2) = \{g_1, g_2, g_3\}$ and $A^{1}(3) = \{g_4, g_5\}$. We update $\pi $ as $\pi(2) = \{g_1, g_2, g_3\}$ and $\pi(3) = \{g_4, g_5\}$ (see Figure~\ref{fig: sma(1, pi)}). 

\begin{figure}[H]
    \centering
    \begin{tikzpicture}[scale=0.85,
            level distance=12mm,
            every node/.style={rectangle, rounded corners=2pt, draw, inner sep=5pt},
            level 1/.style={sibling distance=43mm},
            level 2/.style={sibling distance=21mm},
            level 3/.style={sibling distance=6mm},
            edge from parent/.style={draw, -stealth}
        ]
        
        \node (1) {1}
            child { node (2) {2}
                child { node (4) {4} }
                child { node (5) {5} }
            }
            child { node (3) {3}
                child { node (6) {6} }
                child { node (7) {7} }
            };
    
    \node[draw=none, text=black, above=0.1cm of 1] 
        {$\pi(1) = \{g_1, g_2, g_3, g_4, g_5\}$};
    
    \node[draw=none, text=black, left=0.1cm of 2] 
        {$\pi(2) = \{g_1, g_2, g_3\}$};
    
    \node[draw=none, text=black, right=0.1cm of 3] 
        {$\pi(3) = \{g_4, g_5\}$};

    \node[draw=none, text=black, below=0.1cm of 4] 
        {$\pi(4) = \emptyset$};
    
    \node[draw=none, text=black, below=0.1cm of 5] 
        {$\pi(5) = \emptyset$};

    \node[draw=none, text=black, below=0.1cm of 6] 
        {$\pi(6) = \emptyset$};

    \node[draw=none, text=black, below=0.1cm of 7] 
        {$\pi(7) = \emptyset$};
    \end{tikzpicture}
    \caption{The multilevel allocation $\pi$ after \textbf{SMA}(1, $\pi$, $v$, $\Psi$).}
    \label{fig: sma(1, pi)}
\end{figure}

Then {\bf SMA}$(2, \pi, v, \Psi)$ computes a local allocation $A^2 \in \allocations^{\pi(2)}_{\children(2)}$ that is utilitarian-optimal and Lorenz-dominant w.r.t. $\hat{v}$, e.g., $A^2(4) = \{g_1, g_3\}$ and $A^2(5) = \{g_2\}$, and we set $\pi(4) = \{g_1, g_3\}$ and $\pi(5) = \{g_2\}$ (see Figure~\ref{fig: sma(2, pi)}).

\begin{figure}[H]
    \centering
    \begin{tikzpicture}[scale=0.85,
            level distance=12mm,
            every node/.style={rectangle, rounded corners=2pt, draw, inner sep=5pt},
            level 1/.style={sibling distance=45mm},
            level 2/.style={sibling distance=25mm},
            level 3/.style={sibling distance=6mm},
            edge from parent/.style={draw, -stealth}
        ]
        
        \node (1) {1}
            child { node (2) {2}
                child { node (4) {4} }
                child { node (5) {5} }
            }
            child { node (3) {3}
                child { node (6) {6} }
                child { node (7) {7} }
            };
    
    \node[draw=none, text=black, above=0.1cm of 1] 
        {$\pi(1) = \{g_1, g_2, g_3, g_4, g_5\}$};
    
    \node[draw=none, text=black, left=0.1cm of 2] 
        {$\pi(2) = \{g_1, g_2, g_3\}$};
    
    \node[draw=none, text=black, right=0.1cm of 3] 
        {$\pi(3) = \{g_4, g_5\}$};

    \node[draw=none, text=black, below=0.1cm of 4] 
        {$\pi(4) = \{g_1, g_3\}$};
    
    \node[draw=none, text=black, below=0.1cm of 5] 
        {$\pi(5) = \{g_2\}$};

    \node[draw=none, text=black, below=0.1cm of 6] 
        {$\pi(6) = \emptyset$};

    \node[draw=none, text=black, below=0.1cm of 7] 
        {$\pi(7) = \emptyset$};
    \end{tikzpicture}
    \caption{The multilevel allocation  $\pi$ after \textbf{SMA}(2, $\pi$, $v$, $\Psi$).}
    \label{fig: sma(2, pi)}
\end{figure}

\textbf{SMA}$(3, \pi, v, \Psi)$ computes a local allocation $A^3 \in \allocations^{\pi(3)}_{\children(3)}$ with the same properties. Since node $6$ derives no value from $g_4$ and $g_5$, we get:  $A^3(6) = \emptyset$ and $A^3(7) = \{g_4, g_5\}$. We update $\pi$ accordingly (see Figure~\ref{fig: sma(3, pi)}). 

\begin{figure}[H]
    \centering
    \begin{tikzpicture}[scale=0.85,
            level distance=12mm,
            every node/.style={rectangle, rounded corners=2pt, draw, inner sep=5pt},
            level 1/.style={sibling distance=45mm},
            level 2/.style={sibling distance=25mm},
            level 3/.style={sibling distance=6mm},
            edge from parent/.style={draw, -stealth}
        ]
        
        \node (1) {1}
            child { node (2) {2}
                child { node (4) {4} }
                child { node (5) {5} }
            }
            child { node (3) {3}
                child { node (6) {6} }
                child { node (7) {7} }
            };
    
    \node[draw=none, text=black, above=0.1cm of 1] 
        {$\pi(1) = \{g_1, g_2, g_3, g_4, g_5\}$};
    
    \node[draw=none, text=black, left=0.1cm of 2] 
        {$\pi(2) = \{g_1, g_2, g_3\}$};
    
    \node[draw=none, text=black, right=0.1cm of 3] 
        {$\pi(3) = \{g_4, g_5\}$};

    \node[draw=none, text=black, below=0.1cm of 4] 
        {$\pi(4) = \{g_1, g_3\}$};
    
    \node[draw=none, text=black, below=0.1cm of 5] 
        {$\pi(5) = \{g_2\}$};

    \node[draw=none, text=black, below=0.1cm of 6] 
        {$\pi(6) = \emptyset$};

    \node[draw=none, text=black, below=0.1cm of 7] 
        {$\pi(7) = \{g_4, g_5\}$};
    \end{tikzpicture}
    \caption{The multilevel allocation  $\pi$ after \textbf{SMA}(3, $\pi$, $v$, $\Psi$).}
    \label{fig: sma(3, pi)}
\end{figure}

Since leaves do not recurse, the final allocation is:  $\pi(1) = \items$, $\pi(2)  =  \{g_1, g_2, g_3\}$,  $\pi(3)  =  \{g_4, g_5\}$, $\pi(4) = \{g_1, g_3\}$, $\pi(5)  = \{g_2\}$, $\pi(6) = \emptyset$ and $\pi(7) = \{g_4, g_5\}$.

\end{example}

\noindent Now to prove that {\bf SMA} can be implemented in polynomial time, we first need to show that estimated utilities $\hat{v}_i$, for $i \in \nodes$, are MRF. In the proof, we discuss non-redundant local allocations w.r.t. $\hat{v}$. We define formally this notion below.

\begin{definition}[non-redundancy w.r.t. $\hat{v}$]
    Given a set of nodes $N \subseteq \nodes$, and a bundle $S \subseteq \items$, local allocation $A \in \allocations^S_N$ is non-redundant w.r.t. $\hat{v}$ if $\hat{v}_i(A(i)) - \hat{v}_i(A(i) \setminus \{g\}) > 0$ for all $i \in N$ and all $g \in A(i)$.
\end{definition}

\begin{proposition} \label{proposition: hat v MRF}
Function $\hat{v}_i$ is a matroid rank function for every node $i \in \nodes$.    
\end{proposition}

\begin{proof} For any leaf $x\in \leaves$, we have $\hat{v}_x(S)=u_x(S)$ which is a MRF by definition. For any internal node $i \in \internal$, recall that $\hat{v}_i(S) = \max_{A \in \allocations^S_{\leaves(i)}} \sum_{x \in \leaves(i)} u_x(A(x))$ for any bundle $S \subseteq \items$. Denote $x_1, \ldots, x_\ell$ the leaves in $\leaves(i)$. We know that all leaves in $\leaves(i)$ are equipped with matroid-rank valuations. More precisely, each leaf $x_k \in \leaves(i)$ comes with an associated matroid $\mathcal{M}_k = (\items, \mathcal{J}_k)$ where $\items$ is the ground set of the matroid, and $\mathcal{J}_k$ is the independent sets of the matroid (which is composed of all the bundles $S \subseteq \items$ such that $u_k(S)=|S|$, i.e. all items in $S$ have positive marginal utility).  
Let us consider the union of those matroids, denoted by $\mathcal{M} = (\items, \mathcal{J})$, with $\mathcal{J} = \{J_1 \cup J_2 \cup \cdots \cup J_\ell \ | \ J_k \in \mathcal{J}_k, 1 \leq k \leq \ell\}$. By definition, any set $J \in \mathcal{J}$ can be obtained from $\ell$ (possibly overlapping) bundles such that each bundle belongs to the independent sets of a different leaf  (and $J$ can be obtained by different collections of such independent sets). Note that, it is known that the union of matroids is itself a matroid \cite{NashWilliams1966}. We denote $r_{\mathcal{M}}$ the rank function of $\mathcal{M}$, defined as $r_{\mathcal{M}}(S) = \max\{|J| : J \subseteq S \mbox{ and }  J \in \mathcal{J}\}$ for any $S \subseteq \items$.
We now show that $r_{\mathcal{M}}$ coincides with $\hat{v}_i$.

Recall that, given a bundle $S \subseteq \items$, $\hat{v}_i(S)$ is the welfare of the utilitarian-optimal allocation $A^* \in \allocations^S_{\leaves(i)}$ (by definition). Note that maximizing over all local allocations is equivalent to maximizing over all non-redundant local allocations w.r.t. $\hat{v}_i$. Indeed, from any utilitarian-optimal allocation, one can obtain a non-redundant allocation with the same welfare simply by removing items with zero marginal utility. Let $A$ denote the non-redundant allocation obtained from $A^*$ in this manner. Since $A$
 is non-redundant, we know that $A(x_k)$ is an independent set of matroid $\mathcal{M}_k$ for any $k \in \{1, \ldots, \ell\}$. Hence, by taking the union of the bundles received by the leaves, we obtain an independent set $J$ of $\mathcal{M}$ (by definition of the union of matroids). The cardinality of this set is identical to the utilitarian social welfare of $A$, since the bundles are disjoint and consist of items with positive marginal utilities.

Having established one direction, we now prove the converse. More precisely, we now prove that any set in $\mathcal{J}$ can be associated to some non-redundant allocation w.r.t. $\hat{v}_i$. Indeed, any $J \in \mathcal{J}$,  with $J = J_1 \cup J_2 \cup \cdots \cup J_\ell$, can alternatively be written as a union of disjoint sets: $J = J'_1 \cup J'_2 \cup \cdots \cup J'_\ell$ with $J'_k = J_k \setminus \cup_{l \in [1, k-1]} J_{l}$ for $k \in \{1,\ldots,\ell\}$. This follows from the hereditary property of matroids, which states that, for any matroid $\mathcal{M}_k = (\items, \mathcal{J}_k)$, if $J_k \in \mathcal{J}_k$, then any subset $J_k' \subseteq J_k$ also belongs to $\mathcal{J}_k$. The local allocation $A'$ defined by $A'(x_k) = J'_k$ for all $k \in \{1,\ldots,\ell\}$ is a non-redundant allocation (since $J'_k$ belongs to $\mathcal{M}_k$) whose social welfare is equal to the cardinality of $J$. Hence, the rank of matroid $\mathcal{M}$ coincides with a utilitarian-maximizing non-redundant allocation w.r.t. $\hat{v}$, i.e. $\hat{v}_i$ is a matroid-rank valuation. 
\end{proof}

We next show that \textbf{SMA} has polynomial time complexity. 

\begin{theorem} \label{theorem: sma terminates}
Algorithm {\bf SMA} can be implemented to run in polynomial time with respect to $n$ and $m$, when $\Psi_i \in \fair$ for all $i \in \internal$.
\end{theorem}

\begin{proof}
For any internal node $i\in \internal$, since $\hat{v}_j$ is a MRF for all $j\in \children(i)$ (by Proposition \ref{proposition: hat v MRF}), the local allocation $A^i$ exists and can be computed using existing algorithms designed for standard (monolevel) allocation problems with MRFs. For Lorenz-dominance, weighted leximin, maximum weighed Nash welfare and maximum weighted p-means welfare, the algorithms proposed in \cite{Babaioff_et_al2021, Viswanathan_Zick2023b} apply and run in polynomial time w.r.t. the number of agents (here, $|\children(i)| \le n$), the number of items (here, $|\pi(i)| \le m$), and the complexity of computing the agents' utilities (i.e., the $\hat{v}_j$'s for $j \in \children(i)$). Note that values $\hat{v}_j(\cdot)$ can be computed in polynomial-time by solving a matroid intersection problem as identified by \cite{Benabbou_et_al2021}. 
This can be done using the state-of-the-art algorithm from \cite{Chakrabarty_et_al2019} in $O(m^{3/2} n \log mn \cdot T_u)$, 
where $T_u$ denote the worst-case time to evaluate a leaf’s utility. Since in our setting we assume that values $u_x(\cdot)$ can be accessed in polynomial time w.r.t. $m$, it follows that $A^i$ can be computed in polynomial time  w.r.t. $n$ and $m$ for any $i \in \internal$. 
Finally, since the number of calls to {\bf SMA} is equal to $|\internal|$ and $|\internal| \le n$, we conclude that $\pi^*$ can be constructed in polynomial time w.r.t. $n$ and $m$.

Although the result has just been established, we now show that the values $\hat{v}_j(\cdot)$ can be computed by solving a matroid partition problem rather than a matroid intersection problem, which yields a better theoretical running time. More precisely, using the state-of-the-art algorithm for solving the matroid partition problem \cite{Terao2025} would lead to a complexity of $O((m \log m + n)\sqrt{m} \cdot T_u)$, which is lower than that of matroid intersection. Note also that replacing the matroid intersection step in the algorithms of \cite{Benabbou_et_al2021} and \cite{Babaioff_et_al2021} with a matroid partition algorithm improves the overall computational complexity of their methods. 

We now show that the computation of $\hat{v}_j(S)$ for any bundle $S \subseteq \items$ reduces to solving a matroid partition problem. First, let us define the matroid partition problem: given a collection of $p$ matroids $\mathcal{M}_1 = (E, \mathcal{J}_1), \ldots, \mathcal{M}_p = (E, \mathcal{J}_p)$ defined on a common ground set $E$, the objective is to find a partitionable set $T \subseteq E$ of maximum cardinality; a set $T$ is said to be partitionable if there exists a partition $(T_1, \ldots, T_p)$ such that $T_k \in \mathcal{J}_k$ for any $k \in \{1, \ldots, p\}$. Then, recall that $\hat{v}_j(S) = \max_{A \in \allocations^S_{\leaves(j)}} \sum_{x \in \leaves(j)} u_x(A(x))$ by definition. As  explained in the proof of Proposition~\ref{proposition: hat v MRF}, it suffices to restrict attention to non-redundant local allocations when optimizing. Now let $x_1, \ldots, x_\ell$ denote the leaves in $\leaves(j)$. Since the utilities of the leaves are matroid rank functions (by definition), each leaf $x_k \in \leaves(j)$ comes with an associated matroid $\mathcal{M}_k = (\items, \mathcal{J}_k)$ where $\mathcal{J}_k$ is composed of all bundles whose items have positive marginal utility for $x_k$. Let us consider the restriction of $\mathcal{M}_k$ to $S$, defined as $\mathcal{M}'_k = (S,\mathcal{J}'_k)$ where $\mathcal{J}'_k = \{J \in \mathcal{J}_k \ | \ J \subseteq S\}$. 

We now show that solving the matroid partition problem for the collection of matroids $\mathcal{M}_1', \ldots, \mathcal{M}_\ell'$ amounts to finding a non-redundant allocation that achieves the value $\hat{v}_j(S)$, i.e. that maximizes the sum of the leaves' utilities. First, note that we can associate every non-redundant allocation $A$ to a partitionable set: $T=\cup_{k=1}^\ell A(x_k)$. Indeed, bundles $A(x_k)$, with $k \in \{1, \ldots, \ell\}$, are disjoint and belong to $\mathcal{J}'_k$ by non-redundancy. Moreover we know that $\sum_{x \in \leaves(j)} u_x(A(x)) = |T|$, again by non-redundancy. Conversely, using similar arguments, any partitionable set $T$, with partition $(T_1,\ldots, T_\ell)$, can be associated with a non-redundant allocation $A\in \allocations^S_{\leaves(j)}$, defined by $A(x_k) = T_k$ for all $k \in \{1, \ldots, \ell\}$, such that $\sum_{x \in \leaves(j)} u_x(A(x)) = |T|$. Hence, $\hat{v}_j(S)$ coincide with the largest partitionable set, which establishes the result. 
\end{proof}

Let $\pi^*$ be the multilevel allocation returned by {\bf SMA}. To show that $\pi^*$ is both multilevel utilitarian optimal and multilevel $\Psi$-maximizing w.r.t. $v$, let us remark that this holds w.r.t. $\hat{v}$. 

\begin{lemma} \label{corollary: SMA usw psi}
 $\pi^*$ is both multilevel utilitarian optimal and multilevel $\Psi$-maximizing w.r.t. $\hat{v}$.
\end{lemma}

\begin{proof} 
By Line 5, we have $\pi^*|_{\children(i)} = A^i$ for all nodes $i \in \internal$. Since $A^i$ is both utilitarian optimal and $\Psi_i$-maximizing w.r.t. $\hat{v}$ for $S = \pi(i)$, $\pi^*$ is multilevel utilitarian optimal and multilevel $\Psi$-maximizing w.r.t. $\hat{v}$ (by Def. \ref{def: multilevel usw v hat} and \ref{def: multilevel psi v hat}). 
\end{proof}

The following key proposition relates $\hat{v}$ and $v$ in order to extend this result to the original valuations. The proof can be found in Appendix~\ref{appendix: missing proofs section 3}. 

\begin{restatable}{proposition}{vhatsupv} \label{proposition:v hat  = and sup v}
    For any node $i \in \nodes$, $\hat{v}_i(\pi^*(i)) =  v_i(\pi^*)$, and $\hat{v}_i(\pi(i)) \geq  v_i(\pi)$ for any multilevel allocation $\pi \in \mallocations$.
\end{restatable}

\begin{theorem} \label{theorem: sma multilevel fair}
    {\bf SMA} returns a multilevel allocation that is both multilevel utilitarian optimal and multilevel $\Psi$-maximizing w.r.t. $v$.
\end{theorem}

\begin{proof}
 By Def. \ref{def: multilevel usw} and \ref{def: multilevel psi}, we just need to prove that $\pi^* \vert_{\children(i)}$ is both utilitarian-optimal and $\Psi_i$-maximizing w.r.t. $v$ for any $i \in \internal$. First, recall that $\pi^* \vert_{\children(i)}$ is utilitarian optimal w.r.t. $\hat{v}$ for $S = \pi(i)$ (see Lemma~\ref{corollary: SMA usw psi}). Hence, $\sum_{j \in \children(i)} \hat{v}_j(\pi^*(j)) \geq \sum_{j \in \children(i)} \hat{v}_j(A(j))$ holds for all $A \in \allocations^{\pi^*(i)}_{\children(i)}$. Then, using the first part of Prop. \ref{proposition:v hat  = and sup v}, we get $\sum_{j \in \children(i)} v_j(\pi^*) \geq \sum_{j \in \children(i)} \hat{v}_j(A(j))$ which trivially implies $\sum_{j \in \children(i)} v_j(\pi^*) \geq  \sum_{j \in \children(i)} \hat{v}_j(\pi(j))$ for any $\pi \in  \mallocations$ such that $\pi(i) = \pi^*(i)$. 
 We can then derive $\sum_{j \in \children(i)} v_j(\pi^*) \geq \sum_{j \in \children(i)} \hat{v}_j(\pi(j)) \ge \sum_{j \in \children(i)} {v}_j(\pi(j))$ from the second part of Prop. \ref{proposition:v hat  = and sup v}, showing that $\pi^* \vert_{\children(i)}$ is utilitarian optimal w.r.t. $v$. 
 We now show that $\pi^* \vert_{\children(i)}$ is $\Psi_i$-maximizing w.r.t. $v$. Since $\pi^* \vert_{\children(i)}$ is $\Psi_i$-maximizing w.r.t. $\hat{v}$ for $S = \pi(i)$ (by Lemma~\ref{corollary: SMA usw psi}), $\vv{\hat{v}_i(\pi^*\vert_{\children(i)})} \succeq_{\Psi_i} \vv{\hat{v}_i(A)}$ holds for all $A \in \allocations^{\pi^*(i)}_{\children(i)}$. By Prop.~\ref{proposition:v hat  = and sup v}, $\hat{v}_j(\pi^*(j)) =  v_j(\pi^*)$ holds for any $j \in \children(i)$, and so $\vv{\hat{v}_i(\pi^*\vert_{\children(i)})}  = \vv{v_i(\pi^*\vert_{\children(i)})}$. Thus for any $A \in \allocations^{\pi^*(i)}_{\children(i)}$, we have $\vv{v_i(\pi^*\vert_{\children(i)})} \succeq_{\Psi_i} \vv{\hat{v}_i(A)}$, implying $\vv{v_i(\pi^*\vert_{\children(i)})} \succeq_{\Psi_i} \vv{\hat{v}_i(\pi\vert_{\children(i)})}$ for any $\pi \in \mallocations$ such that $\pi(i) = \pi^*(i)$. 
Similarly, using the second part of Prop.~\ref{proposition:v hat  = and sup v}, we get $\vv{\hat{v}_i(\pi\vert_{\children(i)})} \succeq_{\Psi_i} \vv{{v}_i(\pi\vert_{\children(i)})}$. We then obtain $\vv{v_i(\pi^*\vert_{\children(i)})} \succeq_{\Psi_i} \vv{{v}_i(\pi\vert_{\children(i)})}$ by transitivity. 
\end{proof}

{\bf SMA} is general in the sense that it can be implemented using any algorithm designed for the standard (monolevel) setting to compute $A^i$ at node $i  \in \internal$, and it works with other fairness criteria, provided that there exists a utilitarian optimal allocation that is also fair. For example, we show in Appendix~\ref{appendix: on EF1} that \textbf{SMA} can compute a multilevel envy-free up to one good allocation, using the algorithmn proposed in \cite{Benabbou_et_al2021} and a suitable adaptation of envy-freeness to our multilevel setting. However, although \textbf{SMA} can be implemented in polynomial time for various fairness criteria, it remains computationally expensive in practice, as we show in Section~\ref{section: simulations}. 

\section{Multilevel General Yankee Swap} \label{section: mgys}

In this section, we propose an extension of the General Yankee Swap (\textbf{GYS}) \cite{Viswanathan_Zick2023b} to the multilevel setting, called Multilevel General Yankee Swap ({\bf MGYS}). \textbf{GYS} applies to fairness notions satisfying two properties: (a) fair allocations are Pareto-optimal; (b) the fairness criterion has a coherent \emph{gain function}, non-increasing in an agent's utility, which indicates which agent should be prioritized. The four fairness notions considered in the paper meet these criteria (namely, Lorenz-dominance, weighted leximin, maximum weighted Nash welfare, and weighted p-means welfare). 
While our multilevel extension of \textbf{GYS} always produces a multilevel utilitarian-optimal allocation w.r.t. $v$, it is not necessarily multilevel $\Psi$-maximizing and thus serves as a heuristic. 
Nonetheless, as Section~\ref{section: simulations} shows, it runs significantly faster than {\bf SMA} allowing us to solve much larger instances, while achieving strong fairness in practice. 

\medskip

Basically, {\bf MGYS} amounts to applying {\bf GYS} to the leaves of the tree, with two modifications: (1) we propose a new criterion, adapted for the multilevel setting, to select the leaf for improvement, and (2) we adapt the item transfers between agents to maintain a coherent multilevel allocation. We now describe \textbf{MGYS}, highlighting the differences with \textbf{GYS}.
We first modify $\tree$ by adding an additional leaf, denoted $0$, whose parent is $1$. Although we do not account for  node $0$ when discussing efficiency and fairness, we sometimes treat it as an agent with utility function $u_0(S) = |S|$. {\bf MGYS} starts with the multilevel allocation $\pi $ which assigns all the items to node $0$. At each iteration, it selects a leaf that can either take a new item from node $0$ or steal an item from another leaf under specific conditions. The selection protocol and the allocation step are described below.

\medskip

\noindent \textbf{Exchange graph and item transfers.} The \emph{exchange graph} of a given multilevel allocation $\pi \in \mallocations$, denoted $G(\pi)$, is a directed graph over the set of items $\items$. More precisely, there is an arrow from good $g  \in \items$, owned by leaf $x \in \leaves$, to good $g' \in \items$, owned by another leaf $y \in \leaves \setminus \{x\}$ in $G(\pi)$ if and only if $x$ is indifferent between $g$ and $g'$ given its bundle of items, i.e., $u_x((\pi(x) \setminus \{g\}) \cup \{g'\}) = u_x(\pi(x))$. Given a leaf $x \in \leaves$, let $F_x(\pi)$ be the set of items for which $x$ has a strictly positive marginal gain, i.e., $F_x(\pi) = \{g \in \items : \Delta^{u_x}(\pi(x), g) = 1\}$. For any two leaves $x,y \in \leaves$ with $x \neq y$, a {\em transfer path} is a path from a node in $F_x(\pi)$ to a node in $\pi(y)$.  Given a transfer path $T =  (g_1, g_2, \ldots, g_t)$, we call \emph{path augmentation} the operation consisting in reallocating item $g_t$ to the owner of $g_{t-1}$, $g_{t-1}$ to the owner of $g_{t-2}$, and so on, until $g_1$ is reallocated to leaf $x$. An additional change to {\bf GYS} is that {\bf MGYS} ensures that $\pi $ remains a multilevel allocation after reallocation by appropriately updating the bundles of internal nodes: when a leaf $z$ receives item $g$ from another leaf $z'$, each ancestor $i \in \ancestors(z) \setminus \ancestors(z')$ adds $g$ to its bundle while each ancestor $i \in \ancestors(z') \setminus \ancestors(z)$ removes it.

\medskip

\noindent \textbf{Gain functions and leaf selection.} As in \textbf{GYS}, our algorithm uses gain functions to guide the allocation towards fair outcomes. While, if applied to the leaves, \textbf{GYS} would select the improving leaf by simply maximizing the gain function, {\bf MGYS} instead relies on top-down approach. 
For any node $i\in \nodes \backslash\{1\}$, let $\phi_{\parent(i)}(\pi, i)$ denote the {\em gain function} of agent $i$, capturing the marginal improvement from receiving an item under allocation $\pi $. Note that this function depends on the fairness notion of $\Psi_{\parent(i)} \in \fair$, and returns a $b$-dimensional vector computed from $v_i(\pi)$. Gain functions for the aforementioned fairness criteria are provided in Appendix~\ref{appendix: gain functions}. At each iteration, the leaf $x$ chosen for improvement is selected top-down: starting from the root, we choose the child $i \in \children(1)\backslash \{0\}$ such that $\phi_{1}(\pi, i)$ lexicographically dominates
all others, breaking ties using the least index. 
We repeat the process on the children of $i$, and iterate until reaching a leaf. This leaf is then used for performing a path augmentation, if possible. This process is referred to as $\mathtt{Select\_Leaf}$ and is detailed in Algorithm~\ref{algo:find_leaf}.

\begin{algorithm}[] \small
\caption{\small Select\_Leaf}
\label{algo:find_leaf}
\begin{algorithmic}[1]
    \State \textbf{Input}: $\tree$ - a multilevel tree ; $i \in \nodes$ - a node; $v$ -- valuations of all nodes; $\Psi$ -- fairness criteria of internal nodes.
    \State \textbf{Output} : $x \in \leaves(i)$ - a leaf in $\tree_i$
    \If{$\children(i) = \emptyset$} 
     \State \Return $i$ 
    \EndIf
    \If{$i = 1$} 
    \State $N \gets \arg \underset{j \in \children(i) \setminus \{0\}}{\max} \phi_i(\pi, j)$
    \Else 
    \State $N \gets \arg \underset{j \in \children(i)}{\max} \phi_i(\pi, j)$ \EndIf
    \State $j^* \gets \min_{j \in N} j$
    \State \Return $\mathtt{Select\_Leaf}(\tree, j^*,v,\Psi)$
\end{algorithmic}
\end{algorithm}

\noindent \textbf{Algorithm.} {\bf MGYS} starts with a multilevel allocation $\pi \in \mallocations$ where every item is unallocated, i.e. belongs to node 0. Then, at every iteration, we select a leaf $x \in \leaves$ according to the selection protocol described above. Then, if a transfer path from $F_x(\pi)$ to $\pi(0)$ exists, we find the shortest one and perform the corresponding path augmentation. Otherwise, we remove $x$ from the tree and then prune iteratively any ancestor nodes that become childless. {\bf MGYS} stops when the only nodes remaining are the root and node $0$. Pseudocode is in Algorithm~\ref{algo:multilevel_gys}. In what follows, the multilevel allocation returned by {\bf MGYS} is denoted by $\pi^*$.

\begin{algorithm}[] \small
\caption{\small Multilevel General Yankee Swap (MGYS)}
\label{algo:multilevel_gys}
\begin{algorithmic}[1]
    \State \textbf{Input}: $\tree$ - a multilevel tree ; $\items$ - a set of items; $v$ -- valuations of all nodes ; $\Psi$ -- fairness criteria of internal nodes.
    \State \textbf{Output}: $\pi $ - a multilevel allocation
    \State Set  $\pi(0)= \pi(1)= \items$ and $\pi(x) =  \emptyset$ for all $x \in \nodes \backslash \{0,1\}$.
    \While{$\nodes \neq \{1,0\}$}
    \State $x \gets \mathtt{Select\_Leaf}(\tree, 1,v, \Psi)$
    \State Find the shortest path in $G(\pi)$ from $F_{x}(\pi)$ to $\pi(0).$
    \If{such a path exists}
        \State Perform the corresponding path augmentation.
    \Else
        \State $i \leftarrow x$
        \While{$\children(i) = \emptyset$}
        \State Remove $i$ from $\tree$.
        \State $i\leftarrow \parent(i)$
        \EndWhile
    \EndIf
    \EndWhile
    \State \Return $\pi$
\end{algorithmic}
\end{algorithm}

\begin{example} \label{example: mgys}
We run {\bf MGYS} on the instance from Example~\ref{example: sma}. First, recall that $\Psi_1$ is maximum weighted Nash welfare. Therefore, for any $j \in \children(1)$, the gain function is $\phi_1(\pi, j)  = (1 + 1/v_j(\pi))^{w_j}$ if $v_j(\pi) > 0$, and arbitrarily large otherwise. Since $\Psi_2$ and $\Psi_3$ correspond to Lorenz-dominance, we have $\phi_2(\pi,j) = -v_j(\pi)$ and $\phi_3(\pi,j') = -v_{j'}(\pi)$ for $j \in \children(2)$ and $j' \in  \children(3)$. Initially $\pi(0)= \pi(1) = \items$ (see Figure~\ref{fig: initial pi mgys}). 

\begin{figure}[H]
    \centering
    \begin{tikzpicture}[
            scale=0.85,
            level distance=12mm,
            every node/.style={rectangle, rounded corners=2pt, draw, inner sep=5pt},
            level 1/.style={sibling distance=43mm},
            level 2/.style={sibling distance=21mm},
            level 3/.style={sibling distance=6mm},
            edge from parent/.style={draw, -stealth}
        ]
        
        \node (1) {1}
            child { node (0) {0} }
            child { node (2) {2}
                child { node (4) {4} }
                child { node (5) {5} }
            }
            child { node (3) {3}
                child { node (6) {6} }
                child { node (7) {7} }
            };
    
    \node[draw=none, text=black, above=0.1cm of 1] 
        {$\pi(1) = \{g_1, g_2, g_3, g_4, g_5\}$};

    \node[draw=none, text=black, above=0.1cm of 0] 
        {$\pi(0) = \{g_1, g_2, g_3, g_4, g_5\}$};
    
    \node[draw=none, text=black, left=0.1cm of 2] 
        {$\pi(2) = \emptyset$};
    
    \node[draw=none, text=black, right=0.1cm of 3] 
        {$\pi(3) = \emptyset$};

    \node[draw=none, text=black, below=0.1cm of 4] 
        {$\pi(4) = \emptyset$};
    
    \node[draw=none, text=black, below=0.1cm of 5] 
        {$\pi(5) = \emptyset$};

    \node[draw=none, text=black, below=0.1cm of 6] 
        {$\pi(6) = \emptyset$};

    \node[draw=none, text=black, below=0.1cm of 7] 
        {$\pi(7) = \emptyset$};
    \end{tikzpicture}
    \caption{Initial multilevel allocation  $\pi$.}
    \label{fig: initial pi mgys}
\end{figure}

At the first step, leaf $4$ is selected: node $2$ is prioritized at the root  (since ties are broken using least index), and among its children, leaf $4$ is chosen by the same rule. A path augmentation occurs:  $4$ selects $g_1$. Bundles are modified accordingly: $\pi(4) = \pi(2) = \{g_1\}$ and $\pi(0) = \{g_2,\ldots,g_5\}$ (see Figure~\ref{fig: pi iteration 1}).

\begin{figure}[H]
    \centering
    \begin{tikzpicture}[
            scale=0.85,
            level distance=12mm,
            every node/.style={rectangle, rounded corners=2pt, draw, inner sep=5pt},
            level 1/.style={sibling distance=43mm},
            level 2/.style={sibling distance=21mm},
            level 3/.style={sibling distance=6mm},
            edge from parent/.style={draw, -stealth}
        ]
        
        \node (1) {1}
            child { node (0) {0} }
            child { node (2) {2}
                child { node (4) {4} }
                child { node (5) {5} }
            }
            child { node (3) {3}
                child { node (6) {6} }
                child { node (7) {7} }
            };
    
    \node[draw=none, text=black, above=0.1cm of 1] 
        {$\pi(1) = \{g_1, g_2, g_3, g_4, g_5\}$};

    \node[draw=none, text=black, above=0.1cm of 0] 
        {$\pi(0) = \{g_2, g_3, g_4, g_5\}$};
    
    \node[draw=none, text=black, left=0.1cm of 2] 
        {$\pi(2) = \{g_1\}$};
    
    \node[draw=none, text=black, right=0.1cm of 3] 
        {$\pi(3) = \emptyset$};

    \node[draw=none, text=black, below=0.1cm of 4] 
        {$\pi(4) = \{g_1\}$};
    
    \node[draw=none, text=black, below=0.1cm of 5] 
        {$\pi(5) = \emptyset$};

    \node[draw=none, text=black, below=0.1cm of 6] 
        {$\pi(6) = \emptyset$};

    \node[draw=none, text=black, below=0.1cm of 7] 
        {$\pi(7) = \emptyset$};
    \end{tikzpicture}
    \caption{The multilevel allocation  $\pi$ after iteration 1 of {\bf MGYS}.}
    \label{fig: pi iteration 1}
\end{figure}

At the second step, leaf $6$ is selected: $\phi_1(\pi,3) > \phi_1(\pi,2)$, and among the children of $3$, leaf $6$ is selected by the tie-breaking rule. A path augmentation then occurs: $6$ picks $g_2$, which updates $\pi$ as follows: $\pi(6) = \pi(3) = \{g_2\}$ and $\pi(0) = \{g_3,g_4,g_5\}$.

        
    

    
    

    



At the next iteration, leaf $5$ is selected: $\phi_1(\pi,2) > \phi_1(\pi,3)$ and $\phi_2(\pi,5) > \phi_2(\pi,4)$. It chooses $g_3$, leading to the following updates: $\pi(5) = \{g_3\}$, $\pi(2) = \{g_1,g_3\}$, $\pi(0) = \{g_4,g_5\}$.

        
    

    
    

    



Then $4$ is selected: $\phi_1(\pi,2) > \phi_1(\pi,3)$ and $4$ is chosen by the tie-breaking rule. As no transfer path exists from $4$ to $\pi(0)$, leaf $4$ is removed from $\tree$. The same happens with $5$ at the next step. Then, node $2$ is removed from the tree at the same time as $5$, since it has no remaining children (see Figure~\ref{fig: pi iterations 4 and 5}).

\begin{figure}[H]
    \centering
    \begin{tikzpicture}[
            scale=0.85,
            level distance=12mm,
            every node/.style={rectangle, rounded corners=2pt, draw, inner sep=5pt},
            level 1/.style={sibling distance=48mm},
            level 2/.style={sibling distance=25mm},
            level 3/.style={sibling distance=6mm},
            edge from parent/.style={draw, -stealth}
        ]
        
        \node (1) {1}
            child { node (0) {0} }
            child { node (2) [fill=gray!30] {2}
                child { node (4) [fill=gray!30] {4} }
                child { node (5) [fill=gray!30] {5} }
            }
            child { node (3) {3}
                child { node (6) {6} }
                child { node (7) {7} }
            };
    
    \node[draw=none, text=black, above=0.1cm of 1] 
        {$\pi(1) = \{g_1, g_2, g_3, g_4, g_5\}$};

    \node[draw=none, text=black, above=0.1cm of 0] 
        {$\pi(0) = \{g_4, g_5\}$};
    
    \node[draw=none, text=black, left=0.1cm of 2] 
        {$\pi(2) = \{g_1, g_3\}$};
    
    \node[draw=none, text=black, right=0.1cm of 3] 
        {$\pi(3) = \{g_2\}$};

    \node[draw=none, text=black, below=0.1cm of 4] 
        {$\pi(4) = \{g_1\}$};
    
    \node[draw=none, text=black, below=0.1cm of 5] 
        {$\pi(5) = \{g_3\}$};

    \node[draw=none, text=black, below=0.1cm of 6] 
        {$\pi(6) = \{g_2\}$};

    \node[draw=none, text=black, below=0.1cm of 7] 
        {$\pi(7) = \emptyset$};
    \end{tikzpicture}
    \caption{The multilevel allocation  $\pi$ after iteration 5 of {\bf MGYS}. Deleted nodes are shown in grey.}
    \label{fig: pi iterations 4 and 5}
\end{figure}

Leaf $7$ is then selected (as $\phi_3(\pi,7) >  \phi_3(\pi,6)$), and picks $g_4$: $\pi(7) = \{g_4\}$, $\pi(3) = \{g_2, g_4\}$, $\pi(0) = \{g_5\}$. Next, leaf $6$ is selected by the tie-breaking rule and removed from $\tree$ as no transfer path exists. Finally, $7$ is picked and selects $g_5$. {\bf MGYS} terminates after removing nodes $7$ and $3$, sine nodes $0$ and $1$ are the only remaining ones (see Figure~\ref{fig: pi iterations 8 and 9}). The multilevel allocation returned is the following: $\pi^*(1) = \items$, $\pi^*(2)=\{g_1, g_3\}$,  $\pi^*(3)=\{g_2, g_4, g_5\}$, $\pi^*(4)=\{g_1\}$, $\pi^*(5)=\{g_3\}$, $\pi^*(6)=\{g_2\}$, and $\pi^*(7)= \{g_4, g_5\}$.

\begin{figure}[H]
    \centering
    \begin{tikzpicture}[
            scale=0.85,
            level distance=12mm,
            every node/.style={rectangle, rounded corners=2pt, draw, inner sep=5pt},
            level 1/.style={sibling distance=48mm},
            level 2/.style={sibling distance=25mm},
            level 3/.style={sibling distance=6mm},
            edge from parent/.style={draw, -stealth}
        ]
        
        \node (1) {1}
            child { node (0) {0} }
            child { node (2) [fill=gray!30] {2}
                child { node (4) [fill=gray!30] {4} }
                child { node (5) [fill=gray!30] {5} }
            }
            child { node (3) [fill=gray!30] {3}
                child { node (6) [fill=gray!30] {6} }
                child { node (7) [fill=gray!30] {7} }
            };
    
    \node[draw=none, text=black, above=0.1cm of 1] 
        {$\pi(1) = \{g_1, g_2, g_3, g_4, g_5\}$};

    \node[draw=none, text=black, above=0.1cm of 0] 
        {$\pi(0) = \emptyset$};
    
    \node[draw=none, text=black, left=0.1cm of 2] 
        {$\pi(2) = \{g_1, g_3\}$};
    
    \node[draw=none, text=black, right=0.1cm of 3] 
        {$\pi(3) = \{g_2, g_4, g_5\}$};

    \node[draw=none, text=black, below=0.1cm of 4] 
        {$\pi(4) = \{g_1\}$};
    
    \node[draw=none, text=black, below=0.1cm of 5] 
        {$\pi(5) = \{g_3\}$};

    \node[draw=none, text=black, below=0.1cm of 6] 
        {$\pi(6) = \{g_2\}$};

    \node[draw=none, text=black, below=0.1cm of 7] 
        {$\pi(7) = \{g_4, g_5\}$};
    \end{tikzpicture}
    \caption{The multilevel allocation $\pi$ after iteration 9 of {\bf MGYS}. Deleted nodes are shown in grey.}
    \label{fig: pi iterations 8 and 9}
\end{figure}

\end{example}

Note that {\bf MGYS} terminates since, at each step, either nodes are deleted from $\tree$ due to the non-existence of relevant transfer paths or an item is removed from $\pi(0)$ by path augmentation, thereby reducing the number of possible transfer paths (see lines 7-15). In particular, if $\pi(0)=\emptyset$, we know that no transfer paths remain, and the algorithm deletes nodes until the while-loop condition is satisfied (see line 4). Now, before showing that {\bf MGYS} returns a multilevel utilitarian-optimal allocation w.r.t. $v$, we introduce some useful lemmas.

\begin{lemma} \label{lemma: general bundle = union enfants}
    $\pi(i) = \cup_{j \in\children(i)} \pi(j)$ for any non-redundant multilevel allocation $\pi \in \mallocations$ and any node $i \in \internal$.
\end{lemma}

\begin{proof}
For any $\pi \in \mallocations$ and any node $i \in \nodes$, we have $\pi(i) \supseteq \cup_{j \in \children(i)} \pi(j)$ by definition. If $\pi(i) \supset \cup_{j \in \children(i)} \pi(j)$, then it means that there exists $g \in \pi(i)$ that is not allocated to any children of $i$. Moreover, since $v_i(\pi)=\sum_{j \in \children(i)} v_j(\pi)$, we know that the utility of $v_i$ only depends on the goods allocated to its children. Therefore $v_i(\pi) = v_i(\pi_{-g})$ which implies that $\pi$ is not non-redundant.
\end{proof}

\begin{lemma} \label{lemma: multilevel non-redundancy}
 A multilevel allocation $\pi$ is non-redundant if and only if $v_i(\pi)= |\pi(i)| $ for all $i\in \nodes$.
\end{lemma}

\begin{proof} Let $\pi$ be a non-redundant multilevel allocation. We show that $v_i(\pi) = |\pi(i)|$ for all $i \in \nodes$. If $i$ is a leaf, then we know that $u_i$ is a MRF and therefore has binary marginal gains. This directly implies  that $u_i(\pi(i))=v_i(\pi)= |\pi(i)|$. If $i$ is an internal node, then $v_i(\pi)=\sum_{j \in \children(i)} v_j(\pi) = \sum_{x \in \leaves(i)} v_x(\pi) = \sum_{x \in \leaves(i)} |\pi(x)|$ as just shown. Moreover, we know that $\pi(i) = \cup_{j \in \children} \pi(j)$ due to Lemma~\ref{lemma: general bundle = union enfants}, and therefore $\pi(i) = \cup_{x \in \leaves(i)} \pi(x)$ by associativity of the union operator. Since $\pi(x) \cap \pi(x') = \emptyset$ for any $x,x' \in \leaves(i)$ by definition of $\pi$, then we obtain $v_i(\pi)=|\pi(i)|$. 

Now let $\pi$ be a multilevel allocation satisfying $v_i(\pi)= |\pi(i)| $ for all $i\in \nodes$. We prove that $\pi$ is non-redundant. For any leaf $i \in \leaves$, $v_i(\pi) = u_i(\pi(i))$. Since $u_i$ is a MRF then $u_i(\pi(i)) = |\pi(i)|$ implies that all $g \in \pi(i)$ has a marginal gain equal to $1$. Hence, $v_i(\pi) - v_i(\pi_{-g})) > 0$ for any $g \in \pi(i)$. For any internal node $i \in \internal$,  $v_i(\pi)=\sum_{x \in \leaves(i)} v_x(\pi) = \sum_{x \in \leaves(i)} |\pi(x)|$. Since $\pi(i') \supseteq \cup_{j \in \children(i')} \pi(j)$ for any $i' \in \internal$ by definition of $\pi$, then $\pi(i) \supseteq \cup_{x \in \leaves(i)} \pi(x)$ by associativity of the union operator. Since we have just shown that $v_i(\pi)= |\pi(i)| =  \sum_{x \in \leaves(i)} |\pi(x)|$, then we can derive $\pi(i) = \cup_{x \in \leaves(i)} \pi(x)$ from the fact that $\pi(x) \cap \pi(x') = \emptyset$ for any two leaves $x,x' \in \leaves(i)$ (by definition of $\pi$). For any $g \in \pi(i)$, there exists a unique leaf $y \in \leaves(i)$ such that $g \in \pi(y)$, and we have just proved that $v_y(\pi) - v_y(\pi_{-g}) > 0$ necessarily holds. Therefore $v_i(\pi) - v_i(\pi_{-g}) = \sum_{x \in \leaves(i)} v_x(\pi) - \sum_{x \in \leaves(i)} v_x(\pi_{-g}) = v_y(\pi) - v_y(\pi_{-g}) > 0$. Hence, we just proved that any item had a positive marginal gain, which concludes the proof.
\end{proof}

We adapt Lemma 2.2 and Lemma 2.3 from \cite{Viswanathan_Zick2023b} to our multilevel setting:
\begin{lemma} \label{lemma: transfer path}
    Let $\pi, \pi' \in \mallocations$ be two non-redundant multilevel allocations. Given a leaf $x \in \leaves$ such that $|\pi(x)| < |\pi'(x)|$, there exists a path in $G(\pi)$ from $F_{x}(\pi)$ to $\pi(y)$ for some $y \in \leaves$ such that $|\pi(y)| > |\pi'(y)|.$
\end{lemma}

\begin{lemma} \label{lemma: path augmentation}
    Let $\pi \in \mallocations$ be a non-redundant multilevel allocation. Let $T = (g_1, \ldots, g_t)$ be one of the shortest path from $F_{x}(\pi)$ to $\pi(y)$ for some $x, y \in \leaves$. Let $\pi' \in \mallocations$ be the multilevel allocation obtained from $\pi$ by performing the path augmentation corresponding to $T$. Multilevel allocation $\pi '$ is non redundant and verifies for each $i \in \nodes$:
    $$
    |\pi'(i)| =\left\{
      \begin{array}{ll}
        |\pi(i)| + 1 & \mbox{if } i \in \ancestors(x) \cup \{x\} \backslash \ancestors(y) \\
        |\pi(i)| - 1 & \mbox{if } i \in \ancestors(y) \cup \{y\} \backslash \ancestors(x) \\
        |\pi(i)| & \mbox{otherwise }
      \end{array}
    \right.
    $$
\end{lemma}

\begin{restatable}{proposition}{mgysnonredundant} \label{lemma: MGYS non redundant}
    At any iteration of \textbf{MGYS}, multilevel allocation $\pi \in \mallocations$ maintained by the algorithm is non-redundant.
\end{restatable}

\begin{proof}
    At the beginning of the algorithm, $\pi(0)=\pi(1)=\items$ and $\pi(i)=0$ for all $i \in \nodes\backslash\{0,1\}$. Since by definition $u_0(S) = |S|$ for any bundle of items $S \subseteq \items$, then $v_0(\pi) - v_0(\pi_{-g}) = 1 >0$ for all $g \in \pi(0)$. Moreover, for all $g \in \pi(1)$, $v_{1}(\pi) - v_{1}(\pi_{-g}) =  \sum_{j \in \children(1)} v_j(\pi)- v_j(\pi_{-g}) = v_0(\pi) - v_0(\pi_{-g}) = 1 >0$. Therefore, $\pi$ is initially non-redundant. 
    Then, at each iteration step, we only modify $\pi$ through path augmentation. By Lemma~\ref{lemma: path augmentation}, the resulting allocation remains non-redundant.
\end{proof}

We now prove that \textbf{MGYS} guarantees multilevel utilitarian optimality w.r.t. $v$, but does not ensure multilevel fairness.

\begin{theorem}
    {\bf MGYS} returns a multilevel allocation that is multilevel utilitarian optimal w.r.t. $v$.
\end{theorem}

\begin{proof}
    Assume by contradiction that $\pi^*$ is not multilevel utilitarian optimal w.r.t. $v$. Let $\pi' \in \mallocations$ be a non-redundant multilevel utilitarian optimal allocation w.r.t. $v$ that minimizes $\sum_{x \in \leaves} |u_x(\pi^*(x)) - u_x(\pi'(x))|$ among all such allocations. Since $\pi^*$ is non-redundant (by Proposition \ref{lemma: MGYS non redundant}) and not multilevel utilitarian optimal (by hypothesis), $|\pi^*(y)| < |\pi'(y)|$ holds for some $y \in \leaves$ (by Lemma \ref{lemma: multilevel non-redundancy}). Assume $|\pi^*(x)|  \leq   |\pi'(x)|$ holds for all other leaves $x \in \leaves \setminus \{y,0\}$. Since $\pi^*, \pi'$ are non-redundant, we must have $|\pi^*(0)| > |\pi'(0)|$. Let $\pi_0$ be the multilevel allocation at the beginning of the iteration where $y$ is removed from $\tree$. Note that $\pi_0$ is non-redundant by Proposition~\ref{lemma: MGYS non redundant}. By construction: $|\pi_0(y)| = |\pi^*(y)| < |\pi'(y)|$, $|\pi_0(x)| \leq |\pi^*(x)| \leq |\pi'(x)|$ for any $x \in \leaves \setminus \{0\}$ and $|\pi_0(0)| \geq |\pi^*(0)| > |\pi'(0)|$. By Lemma~\ref{lemma: transfer path}, there is a transfer path from $F_y(\pi_0)$ to $\pi_0(0)$ in $G(\pi_0)$, contradicting the definition of $\pi_0$.     Hence, we must have $|\pi^*(y')| > |\pi'(y')|$ for some leaf $y' \in  \leaves \setminus \{0\}$. By Lemma~\ref{lemma: transfer path}, there is a transfer path from $F_{y'}(\pi')$ to $\pi'(z)$ in $G(\pi')$ for some $z \in \leaves$ with $|\pi^*(z)| < |\pi'(z)|$. By Lemma~\ref{lemma: path augmentation}, transferring goods along this path leads to a non-redundant multilevel allocation $\pi_1$ s.t. $|\pi_1(y')| = |\pi'(y')| + 1$, $|\pi_1(z)| = |\pi'(z)| - 1$ and $|\pi_1(x')| = |\pi'(x')|$ for all $x'  \in  \leaves \backslash \{y',z\}$. If $z  =  0$, then $v_{1}(\pi_1) > v_{1}(\pi')$ (since $y'$ increased its utility without reducing any other agent’s), contradicting the root-optimality of $\pi'$. If $z \neq 0$, then we contradict its minimality assumption since $\sum_{x \in \leaves} |u_x(\pi^*(x)) - u_x(\pi'(x))| > \sum_{x \in \leaves} |u_x(\pi^*(x)) - u_x(\pi_1(x))|$.
\end{proof}

\begin{theorem}
    {\bf MGYS} does not necessarily return a multilevel $\Psi$-maximizing allocation w.r.t. $v$, even when $\Psi_i \in \fair$ for all nodes $i \in \internal$.\color{black}
\end{theorem}

\begin{proof} 
Consider Example~\ref{example: sma}, and let $\pi^*_1$ and $\pi^*_2$ be the multilevel allocations returned by \textbf{SMA} and \textbf{MGYS}, respectively (see Examples~\ref{example: sma}-\ref{example: mgys}). Both $\pi^*_1 \vert_{\children(1)}$ and $\pi^*_2 \vert_{\children(1)}$ belong to $\allocations^\items_{\children(1)}$. The weighted Nash welfare is $3^5  \times  2^2  =  972$ for $\pi^*_1 \vert_{\children(1)}$ and only $2^5 \times  3^2  =  288$ for $\pi^*_2 \vert_{\children(1)}$. Thus $\pi^*_2 \vert_{\children(1)}$ is not $\Psi_1$-maximizing w.r.t. $v$. Hence, $\pi^*_2$ is not multilevel $\Psi$-maximizing w.r.t. $v$.
Note that this instance would still constitute a counterexample if $\Psi_1$ was Lorenz dominance, weighted leximin, or weighted p-means welfare (with $p=1$, for instance). 
Similar counterexamples can be obtained in unweighted settings.
\end{proof}

Note that, in \textbf{MGYS}, only the root may leave  items unallocated. Indeed, items are chosen directly by the leaves, and $\pi(i) = \cup_{j \in \children(i)} \pi(j)$ for any internal node $i \in \internal$ by construction. The discarded items are those allocated to node $0$. In {\bf SMA}, when {\bf GYS} is applied at each internal node, items may be discarded only during the initial call at the root (as allowed by {\bf GYS}), but not in subsequent calls. Indeed, at the root, {\bf GYS} may leave items unallocated to ensure non-redundancy, making the root allocation $A^{1}$ non-redundant and utilitarian-optimal with respect to $\hat{v}$. Moreover, for every child $i$ of the root, $\hat{v}_i$ is a MRF (hence has binary marginal gains). Therefore, the non-redundancy of $A^{1}$ implies $\hat{v}_i(A^{1}(i))=|A^{1}(i)|$, and its local allocation $A^i$ must achieve the same value (as it is utilitarian optimal w.r.t. $\hat{v}$).  
As $\hat{v}_j$ has binary marginal gains for every child $j$ of $i$, all items assigned to node $i$ must be allocated to its children (otherwise, total utility would be strictly smaller than $|A^{1}(i)|$). By induction, only the root may leave items unallocated.  

\medskip

We now state the worst-case time complexity of {\bf MGYS}. Let $T_u$ denote the worst-case time to evaluate a leaf’s utility (i.e., for a \emph{given} MRF), and let $b$ be the maximum dimension of gain functions $\phi_i$ for $i \in \internal$. For the fairness notions considered here, $b$ is in $O(1)$.

\begin{theorem} \label{theorem: complexity mgys}
    The worst-case time complexity of {\bf MGYS} is  $O((n + m)(m T_u \log m +  mn+ nb)).$
\end{theorem}

\begin{proof} The algorithm stops whenever nodes $1$ and $0$ are the only remaining nodes. At each iteration step of \textbf{MGYS}, we either transfer one item from node $0$ to a leaf or delete a leaf which has no transfer path to $0$. Note that, when the deletion of a leaf implies that an internal node has no more leaves in its subtree, this internal node is deleted at the same time. Note also that, if $\pi(0)=\emptyset$, then there exist no more transfer path to $0$. Therefore, the number of iterations is upper bounded by $m$ the number of items plus the number of leaves. Hence, the worst case is the monolevel case where every node (but the root) is a leaf, i.e. the number of iterations is in $O(m + n)$. 

Every iteration step starts with the selection of a leaf, which basically only compares gain functions of siblings. As a consequence, the monolevel case is one possible worst case, requiring the comparison of all nodes (but the root). Therefore this operation is in $O(n(b + T_\phi))$, where $T_\phi$ is the worst-case complexity of computing the value of a gain function. Note that $T_\phi$ is in $O(T_u+b)$ in general, but {\bf MGYS} can be implemented so that it is in $O(b)$ as observed in \cite{Viswanathan_Zick2023b}. Therefore the selection step is in $O(nb)$. Then, the algorithm computes the shortest path in the corresponding exchange graph. This operation can be performed using algorithm $\mathtt{Get-Distances}$ presented in \cite{Viswanathan_Zick2023b} which has been proven to run in $O(m T_u \log m)$. If such a path has been found, a path augmentation is performed. More precisely, each item in the path is reallocated from an agent to the one dictated by the transfer path, and their ancestors must update their bundle accordingly. Hence, in the worst case, all the items are involved in the path augmentation, and therefore it is in $O(nm)$. If no transfer path is found, then the algorithm only deletes at most $n$ nodes which can be performed in $O(n)$. Hence, each iteration step runs in $O(m T_u \log m + mn + nb)$.  
\end{proof}

Note that in some cases, \textbf{SMA} may incur much higher computational complexity than \textbf{MGYS}. This is because computing $\hat{v}_i(S)$ for an internal node $i \in \internal$ and a set of $S \subseteq \items$ requires solving a matroid partition problem, which remains relatively expensive. Using the algorithm of \cite{Terao2025}, this takes $O((m_i \log m_i + n_i) \cdot \sqrt{m_i} \cdot T_u)$, where $m_i = |S|$ and $n_i = |\leaves(i)|$. As an illustration, consider an instance where $m > n$, and the root has two children, and apply {\bf SMA} using {\bf GYS} to compute each $A^i$ during the recursive calls. The first call at the root then requires running {\bf GYS} on two agents, where each utility evaluation involves solving a matroid partition problem. This leads to a time complexity in the order of $m^{7/2} (\log m)^2 \cdot T_u + n \sqrt{m}  \log m \cdot T_u + m b$ just for the initial call, which already exceeds the overall complexity of \textbf{MGYS}. This will be illustrated in the experimental section.

However, we also note that \textbf{MGYS} is slightly less general than \textbf{SMA}, as it requires fairness notions for which a coherent gain function can be defined. For instance, defining such a function for EF1 would be challenging, as already observed in \cite{Viswanathan_Zick2023b}.

\section{Experimental comparisons of the Algorithms} \label{section: simulations}

In this section, we compare {\bf SMA} and {\bf MGYS} computation times and evaluate {\bf MGYS} fairness by measuring the fraction of instances where it fails to return a multilevel $\Psi$-maximizing allocation ($err_1$). For these instances, we measure the error as the sum, over all agents, of the absolute differences between their bundle size and that in a locally $\Psi_i$-maximizing allocation ($err_2$). As a baseline, we compare the fairness performance of \textbf{MGYS} to that of \textbf{GYS} applied to the leaves, updating the bundles of the internal nodes accordingly.

\paragraph{Implementation details.} Experiments were conducted on a server equipped with two Intel Xeon X5690 CPUs running at 3.47GHz, and 144 GB of RAM (each experiment was run on 4 cores). The program is written in Python. All results were obtained over 200 instances, with a timeout set to 20 minutes. For the construction of exchange graphs and the computation of shortest paths, we use the NetworkX library. 

\paragraph{Instance generation.} We perform experiments over two types of tree structures: a balanced tree with branching factor 3 and a comb tree (i.e. a path-like tree in which every internal node has exactly one extra child). We randomly generate the leaves MRFs among the following classes: binary additive, binary additive with cardinality constraints, cardinal valuation, and binary assignment. 
Preferences over singletons are drawn independently, with sparsity introduced by sampling item preferences from a Bernoulli distribution with parameter $p$.  Agent weights are sampled uniformly at random from the integer interval $[[1, 5]]$. 
To facilitate reproducibility, we set the random seed to 1 at the beginning of all experiments\footnote{Numerical tests were also run on different seeds, just to make sure our conclusions are robust.}. 

\paragraph{Algorithms.} Regarding the algorithms, {\bf SMA} was run using {\bf GYS} as the local algorithm at each internal node, as it is the fastest available option. Recall that, as a baseline, we considered {\bf GYS} applied to the leaves. In these experiments, we focused on instances where all nodes seek a Lorenz-dominant allocation, so as to be fair to \textbf{GYS}, which considers only one fairness notion at the leaves. Note that using Lorenz dominance also avoids the issue of internal node weights, which is not handled by {\bf GYS} when applied only to the leaves.

\paragraph{Running time.} In the following,  we report both the mean and the standard deviation. We present in Table  \ref{tab: running time} a comparison of the running time of the two proposed algorithms: {\bf SMA} and {\bf MGYS}. These results cover instances varying the number of agents and items independently, jointly, as well as the impact of the hierarchical structure. As we did not notice any impact of the sparsity of preferences on computation times, we only report results for $p = 0.5$.

\begin{table}[h!]
    \small
    \centering
    \caption{Computation time for $p = 0.5$, varying the numbers of agents and items (‘–’ indicates a timeout).}
    \begin{tabular}{|c|c|c|c|c|c|c|c|}
      \hline
      \multirow{3}{*}{n} & \multirow{3}{*}{m} & \multicolumn{2}{c|}{\textbf{SMA}} & \multicolumn{2}{c|}{\textbf{MGYS}} \\
      \cline{3-6}
      & & Balanced Tree & Comb Tree & Balanced Tree & Comb Tree \\
      \hline
      $12$ & $20$  & $3.594 \pm 3.156$ & $4.850 \pm 6.662$ & $0.136 \pm 0.109$ & $0.127 \pm 0.133$ \\
      $12$ & $30$  & $19.553 \pm 20.226$ & $26.854 \pm 37.007$ & $0.415 \pm 0.386$ & $0.464 \pm 0.558$ \\
      $12$ & $40$  & $85.382 \pm 80.399$ & $132.882 \pm 208.716$ & $1.022 \pm 1.002$ & $1.154 \pm 1.351$ \\
      $12$ & $50$  & $226.963 \pm 260.205$ & $377.806 \pm 608.885$ & $2.323 \pm 2.275$ & $2.120 \pm 2.483$ \\
      $12$ & $100$  & - & - & $18.040 \pm 19.425$ & $16.192 \pm 21.620$ \\
      $12$ & $200$  & - & - & $124.808 \pm 124.120$ & $105.755 \pm 129.522$ \\
      $12$ & $300$  & - & - & $404.667 \pm 347.273$ & $319.567 \pm 309.994$ \\
      $12$ & $400$  & - & - & $803.072 \pm 773.169$ & $738.174 \pm 738.054$ \\
      $12$ & $500$  & - & - & - & - \\
      \hline
      $10$ & $20$  & $2.756 \pm 2.878$ & $4.854 \pm 6.212$ & $0.134 \pm 0.121$ & $0.127 \pm 0.133$ \\
      $20$ & $20$  & $3.734 \pm 3.045$ & $4.339 \pm 6.049$ & $0.164 \pm 0.105$ & $0.193 \pm 0.186$ \\
      $30$ & $20$  & $4.184 \pm 3.124$ & $4.537 \pm 5.858$ & $0.215 \pm 0.138$ & $0.213 \pm 0.217$ \\
      \hline
      $20$ & $40$  & $85.834 \pm 73.059$ & $116.489 \pm 167.763$ & $1.137 \pm 0.810$ & $1.240 \pm 1.530$ \\
      $30$ & $60$  & $525.322 \pm 400.514$ & $736.185 \pm 1196.904$ & $3.632 \pm 2.097$ & $4.614 \pm 5.592$ \\
      $40$ & $80$  & - & - & $9.048 \pm 4.412$ & $9.797 \pm 12.093$ \\
      $50$ & $100$  & - & - & $19.674 \pm 8.978$ & $26.006 \pm 28.499$ \\
      $60$ & $120$  & - & - & $34.945 \pm 17.049$ & $38.108 \pm 46.705$ \\
      $70$ & $140$  & - & - & $56.286 \pm 26.013$ & $57.834 \pm 69.690$ \\
      $100$ & $200$  & - & - & $167.116 \pm 87.478$ & $151.913 \pm 166.255$ \\
      $150$ & $300$  & - & - & $553.209 \pm 187.716$ & $533.791 \pm 494.589$ \\
      $200$ & $400$  & - & - & - & $1151.340 \pm 996.705$ \\
      \hline
    \end{tabular}
    \label{tab: running time}
\end{table}

Results show that \textbf{MGYS} is significantly faster than \textbf{SMA}. For instance, on balanced trees with $n=12, m=20$, the average running time of \textbf{SMA} is 3.594 seconds, whereas \textbf{MGYS} runs in 0.136 seconds, making it about 26 times faster than \textbf{SMA}. This performance gap widens as the instance size increases: for $n=12, m=50$, \textbf{MGYS} is nearly a hundred times faster than \textbf{SMA}.  Hence, we also observe that \textbf{MGYS} is highly scalable, successfully solving instances with up to $n = 200$ agents and $m = 400$ items before hitting the time limit. Moreover, we see that \textbf{MGYS} is robust to the hierarchical structure, whereas \textbf{SMA} performs worse on comb trees than on balanced trees.

\paragraph{Fairness performance} We evaluated the empirical fairness of \textbf{MGYS} by measuring the fraction of instances for which it fails to return a multilevel $\Psi$-maximizing allocation w.r.t. $v$. To do this, for each allocation $\pi^* \in \mallocations$ computed by \textbf{MGYS}, we ran the monolevel algorithm \textbf{GYS} at each internal node $i \in \internal$ on the set of items $\pi^*(i)$, and compared its output to that of \textbf{MGYS}. We show that \textbf{MGYS} rarely fails to produce a multilevel $\Psi$-maximizing allocation w.r.t. $v$ ($err1$), and even when it does, the allocation returned is still close to fair (as $err2$ is also very low). We present the complete results in the following. For $err1$, we report a 95\% confidence interval, while for $err2$, we report the mean and standard deviation.

As already said, \textbf{MGYS} rarely fails to construct a multilevel $\Psi$-maximizing allocation w.r.t. $v$, and even under very restrictive preferences ($p = 0.1$), it fails in only $10\%$ to $22\%$ of instances (at the $95\%$ confidence level). We ran the same tests using correlated preferences, generated by first assigning singleton preferences to a reference agent (via a Bernoulli distribution with parameter $p$), then copying these preferences to each other agent with probability $\rho = 0.8$. In Table \ref{tab:preferences_fairness}, we observe that it only has a small impact on the performance of \textbf{MGYS}, confirming the good fairness properties of this heuristic algorithm. We also observe a slight degradation in fairness on comb trees compared to balanced trees (see Tables \ref{tab:preferences_fairness} and \ref{tab:m_size_fairness}). Moreover, fairness performance appears unaffected by instance size, as shown in Table~\ref{tab:m_size_fairness}.

\begin{table}[h!]
\small
\centering
\begin{minipage}{0.48\textwidth}
  \centering
  \caption{Values of error $err_1$ for $n=15$ and $m=25$.}
  \begin{tabular}{|c|c|c|c|}
    \hline
    Preferences & $p$ & Balanced Tree & Comb Tree \\
    \hline
    $Indep$ & $0.1$ & $[0.10, 0.22]$ & $[0.11, 0.19]$ \\
    $Indep$ & $0.5$ & $[0.00, 0.00]$ & $[0.00, 0.04]$ \\
    $Indep$ & $0.9$ & $[0.00, 0.00]$ & $[0.00, 0.00]$ \\
    \hline
    $Corr$ & $0.1$ & $[0.06, 0.14]$ & $[0.07, 0.15]$ \\
    $Corr$ & $0.5$ & $[0.00, 0.04]$ & $[0.11, 0.19]$ \\
    $Corr$ & $0.9$ & $[0.00, 0.00]$ & $[0.01, 0.05]$ \\
    \hline
  \end{tabular}
  \label{tab:preferences_fairness}
\end{minipage}
\hfill
\begin{minipage}{0.48\textwidth}
  \centering
  \caption{Values of error $err_1$ for $p=0.5$.}
  \begin{tabular}{|c|c|c|c|}
    \hline
    $n$ & $m$ & Balanced Tree & Comb Tree \\
    \hline
    $12$ & $20$ & $[0.00, 0.00]$ & $[0.01, 0.05]$ \\
    $12$ & $30$ & $[0.00, 0.02]$ & $[0.01, 0.06]$ \\
    $12$ & $40$ & $[0.00, 0.00]$ & $[0.00, 0.04]$ \\
    \hline
    $10$ & $20$ & $[0.00, 0.04]$ & $[0.01, 0.09]$ \\
    $20$ & $20$ & $[0.00, 0.00]$ & $[0.02, 0.06]$ \\
    $30$ & $20$ & $[0.00, 0.00]$ & $[0.01, 0.05]$ \\
    \hline
  \end{tabular}
  \label{tab:m_size_fairness}
\end{minipage}
\end{table}

We also report results on the distance to a multilevel $\Psi$-maximizing allocation ($err_2$) when \textbf{MGYS} fails to produce one (see Tables \ref{tab:preferences_fairness_distance} and \ref{tab:m_size_fairness_distance}). Values of $err_2$ remain consistently low, indicating that even when \textbf{MGYS} fails to produce a multilevel $\Psi$-maximizing allocation w.r.t. $v$, it still remains close in terms of fairness. For instance, for independent preferences with sparsity parameter $p=0.1$, we observe that  $err2$ is 2.19. Since $err2$ counts each misallocated item twice (once for the node that incorrectly receives it and once for the node that should have received it), this corresponds to approximately one misallocated item on average. As before, we find no significant impact of instance size on this metric, and we observe that this error is slightly higher in comb trees than in balanced trees.

\begin{table}[h!]
\small
\centering
\begin{minipage}{0.48\textwidth}
  \centering
  \caption{Values of $err_2$ for $n=15$ and $m=25$.}
  \begin{tabular}{|c|c|c|c|}
    \hline
    Preferences & $p$ & Balanced Tree & Comb Tree \\
    \hline
    $Indep$ & $0.1$ & $2.19 \pm 0.59$ & $2.28 \pm 0.87$ \\
    $Indep$ & $0.5$ & $0.00$ & $3.00 \pm 1.00$ \\
    $Indep$ & $0.9$ & $0.00$ & $0.00$ \\
    \hline
    $Corr$ & $0.1$ & $2.53 \pm 0.88$ & $2.73 \pm 1.14$ \\
    $Corr$ & $0.5$ & $3.50 \pm 0.87$ & $3.86 \pm 2.10$ \\
    $Corr$ & $0.9$ & $0.00$ & $3.33 \pm 1.49$ \\
    \hline
  \end{tabular}
  \label{tab:preferences_fairness_distance}
\end{minipage}
\hfill
\begin{minipage}{0.48\textwidth}
  \centering
  \caption{ Values of $err_2$ for $p=0.5$.}
  \begin{tabular}{|c|c|c|c|}
    \hline
    $n$ & $m$ & Balanced Tree & Comb Tree \\
    \hline
    $12$ & $20$ & $0.00$ & $2.00 \pm 0.00$ \\
    $12$ & $30$ & $2.00 \pm 0.00$ & $2.25 \pm 0.66$ \\
    $12$ & $40$ & $0.00$ & $3.5 \pm 1.66$ \\
    \hline
    $10$ & $20$ & $2.00 \pm 0.00$ & $2.20 \pm 0.60$ \\
    $20$ & $20$ & $0.00$ & $2.86 \pm 0.99$ \\
    $30$ & $20$ & $0.00$ & $3.00 \pm 1.00$ \\
    \hline
  \end{tabular}
  \label{tab:m_size_fairness_distance}
\end{minipage}
\end{table}

Though \textbf{MGYS} performs very well in terms of fairness, it is natural to ask how existing monolevel algorithms behave in the same setting. To investigate this, we apply \textbf{GYS} to the leaves of the tree and update the bundles of the internal nodes accordingly. We assume that all internal nodes use Lorenz-dominance as their fairness criterion, and therefore impose Lorenz-dominance between the leaves as well when running \textbf{GYS}. We summarize the results in Table~\ref{tab:preferences_fairness_GYS} and Table~\ref{tab:m_size_fairness_GYS}. We observe that \textbf{MGYS} exhibits similar behavior to that observed in the general setting described above. In contrast, \textbf{GYS} on the leaves performs very poorly in terms of fairness. In particular, it almost never yields a multilevel $\Psi$-maximizing allocation with respect to $v$ (since $err1$ is close to $1$), and for comb trees, the resulting allocations can be very far from multilevel fair ones (as indicated by high $err2$ values). These results emphasize that computing a multilevel fair allocation is fundamentally different from computing a fair allocation only at the leaves, which highlight the need for proper multilevel algorithms. They also highlight that the strong fairness properties achieved by \textbf{MGYS} are highly non-trivial.

\begin{table}[h!]
    \centering
    \caption{Values of $err_1$ and $err_2$  for $n=15$ and $m=25$ for \textbf{MGYS} and \textbf{GYS}.}
    \begin{tabular}{|c|c|c||cc||cc|}
        \hline
        \multirow{2}{*}{Pref.} & \multirow{2}{*}{$p$} & \multirow{2}{*}{Algo.} & \multicolumn{2}{c||}{Balanced tree} & \multicolumn{2}{c|}{Comb tree} \\
        \cline{4-7}
        & & & $err_1$ & $err_2$ & $err_1$ & $err_2$ \\
        \hline \hline
        \multirow{2}{*}{$Indep$} & \multirow{2}{*}{$0.1$} 
            & \textbf{MGYS} & $[0.13, 0.25]$ & $2.47 \pm 0.97$ & $[0.12, 0.24]$ & $2.34 \pm 0.89$ \\
        & & \textbf{GYS}  & $[0.75, 0.87]$ & $3.43 \pm 1.58$ & $[0.78, 0.90]$ & $11.31 \pm 8.07$ \\
        \hline
        \multirow{2}{*}{$Indep$} & \multirow{2}{*}{$0.5$} 
            & \textbf{MGYS} & $[0.00, 0.00]$ & $0.00$ & $[0.04, 0.12]$ & $2.80 \pm 1.42$ \\
        & & \textbf{GYS}  & $[0.95, 0.99]$ & $3.19 \pm 1.62$ & $[0.98, 1.00]$ & $26.66 \pm 10.41$ \\
        \hline
        \multirow{2}{*}{$Indep$} & \multirow{2}{*}{$0.9$} 
            & \textbf{MGYS} & $[0.00, 0.00]$ & $0.00$ & $[0.00, 0.00]$ & $0.00$ \\
        & & \textbf{GYS}  & $[0.95, 0.99]$ & $3.04 \pm 1.46$ & $[0.95, 0.99]$ & $27.51 \pm 10.53$ \\
        \hline \hline
        \multirow{2}{*}{$Corr$} & \multirow{2}{*}{$0.1$} 
            & \textbf{MGYS} & $[0.05, 0.13]$ & $2.22 \pm 0.63$ & $[0.09, 0.17]$ & $2.72 \pm 1.25$ \\
        & & \textbf{GYS}  & $[0.91, 0.99]$ & $3.12 \pm 1.49$ & $[0.95, 0.99]$ & $16.92 \pm 9.73$ \\
        \hline
        \multirow{2}{*}{$Corr$} & \multirow{2}{*}{$0.5$} 
            & \textbf{MGYS} & $[0.00, 0.04]$ & $3.00 \pm 1.00$ & $[0.10, 0.22]$ & $3.88 \pm 1.93$ \\
        & & \textbf{GYS}  & $[0.95, 0.99]$ & $3.35 \pm 1.70$ & $[0.97, 1.00]$ & $25.71 \pm 9.45$ \\
        \hline
        \multirow{2}{*}{$Corr$} & \multirow{2}{*}{$0.9$} 
            & \textbf{MGYS} & $[0.00, 0.00]$ & $0.00$ & $[0.00, 0.02]$ & $4.00 \pm 0.00$ \\
        & & \textbf{GYS}  & $[0.95, 0.99]$ & $3.26 \pm 1.62$ & $[0.97, 1.00]$ & $27.56 \pm 10.79$ \\
        \hline
    \end{tabular}
    \label{tab:preferences_fairness_GYS}
\end{table}

\begin{table}[h!]
    \centering
    \caption{Values of $err_1$ and $err_2$ with $p=0.5$ for \textbf{MGYS} and \textbf{GYS}.}
    \begin{tabular}{|c|c|c||cc||cc|}
        \hline
        \multirow{2}{*}{$n$} & \multirow{2}{*}{$m$} & \multirow{2}{*}{Algo.} & \multicolumn{2}{c||}{Balanced tree} & \multicolumn{2}{c|}{Comb tree} \\
        \cline{4-7}
        & & & $err_1$ & $err_2$ & $err_1$ & $err_2$ \\
        \hline \hline
        \multirow{2}{*}{$12$} & \multirow{2}{*}{$20$} 
            & \textbf{MGYS} & $[0.00, 0.00]$ & $0.00$ & $[0.01, 0.05]$ & $2.00 \pm 0.00$ \\
        & & \textbf{GYS}  & $[0.78, 0.90]$ & $5.94 \pm 1.78$ & $[0.96, 1.00]$ & $17.16 \pm 7.91$ \\
        \hline
        \multirow{2}{*}{$12$} & \multirow{2}{*}{$30$} 
            & \textbf{MGYS} & $[0.00, 0.00]$ & $0.00$ & $[0.01, 0.05]$ & $2.67 \pm 0.94$ \\
        & & \textbf{GYS}  & $[0.91, 0.99]$ & $8.45 \pm 3.01$ & $[0.95, 0.99]$ & $24.13 \pm 10.81$ \\
        \hline
        \multirow{2}{*}{$12$} & \multirow{2}{*}{$40$} 
            & \textbf{MGYS} & $[0.00, 0.00]$ & $0.00$ & $[0.01, 0.05]$ & $2.33 \pm 0.75$ \\
        & & \textbf{GYS}  & $[0.90, 0.98]$ & $11.34 \pm 4.71$ & $[0.91, 0.99]$ & $31.98 \pm 15.15$ \\
        \hline \hline
        \multirow{2}{*}{$10$} & \multirow{2}{*}{$20$} 
            & \textbf{MGYS} & $[0.00, 0.00]$ & $0.00$ & $[0.00, 0.04]$ & $2.00 \pm 0.00$ \\
        & & \textbf{GYS}  & $[0.87, 0.95]$ & $4.66 \pm 1.67$ & $[0.90, 0.98]$ & $10.84 \pm 4.74$ \\
        \hline
        \multirow{2}{*}{$20$} & \multirow{2}{*}{$20$} 
            & \textbf{MGYS} & $[0.00, 0.00]$ & $0.00$ & $[0.00, 0.02]$ & $2.00 \pm 0.00$ \\
        & & \textbf{GYS}  & $[0.87, 0.95]$ & $3.54 \pm 2.25$ & $[1.00, 1.00]$ & $41.84 \pm 11.18$ \\
        \hline
        \multirow{2}{*}{$30$} & \multirow{2}{*}{$20$} 
            & \textbf{MGYS} & $[0.00, 0.00]$ & $0.00$ & $[0.00, 0.04]$ & $2.50 \pm 0.87$ \\
        & & \textbf{GYS}  & $[1.00, 1.00]$ & $7.47 \pm 0.99$ & $[1.00, 1.00]$ & $65.93 \pm 11.43$ \\
        \hline
    \end{tabular}
    \label{tab:m_size_fairness_GYS}
\end{table}

\paragraph{On the number of discarded items.} As explained in the paper, our algorithms may leave some items unallocated at the root. In Table~\ref{tab: discarded items}, we show that the number of discarded items is extremely low, and actually very often null. Note that the scores of \textbf{SMA} and \textbf{MGYS} are identical, which is expected since both algorithms are multilevel utilitarian optimal w.r.t. $v$ and we use \textbf{GYS} at each internal node in \textbf{SMA}.

\begin{table}[h!] \small
\centering
\caption{Average number of discarded items for $p=0.5$ for \textbf{SMA} and \textbf{MGYS}.}
\begin{tabular}{|c|c|c|c|c|c|c|c|}
  \hline
  $n$ & $m$ & Balanced Tree & Comb Tree  \\
  \hline
  $12$ & $20$  & $0.01 \pm 0.14$ & $0.00$ \\
  $12$ & $30$  & $0.00$ & $0.16 \pm 1.05$ \\
  $12$ & $40$  & $0.05 \pm 0.42$ & $0.12 \pm 0.74$ \\
  \hline
  $10$ & $20$  & $0.04 \pm 0.31$ & $0.12 \pm 0.52$ \\
  $20$ & $20$  & $0.00$ & $0.00$ \\
  $30$ & $20$  & $0.00$ & $0.00$ \\
  \hline
\end{tabular}
\label{tab: discarded items}
\end{table}

\section{Future Work}

In this paper, we introduced the multilevel fair allocation problem, extending fair allocation to settings with hierarchical structure and thereby considerably broadening its range of applications. We presented two algorithms. The first one, \textbf{SMA}, follows an intuitive top-down approach and satisfies all desired efficiency and fairness properties, but remains relatively expensive to run. This motivated us to propose a multilevel extension of the General Yankee Swap \cite{Viswanathan_Zick2023b}. Although this algorithm loses formal fairness guarantees, it remains fair experimentally and is significantly more efficient than \textbf{SMA}, making it suitable for real-world applications. Finally, we showed that both algorithms could be extended to handle chores.

Numerous questions remain open. One concerns our modeling assumption for the hierarchical structure. We chose to represent the hierarchy as a tree, which is arguably the most natural approach. However, some applications may require more general structures. For instance, one might instead use a directed acyclic graph, as in situations where a laboratory belongs to several departments. Such an extension would have significant implications for the algorithmic design. In particular, a top-down procedure such as the one used in \textbf{SMA} would require substantial adaptations.

Another limitation of our model lies in the assumption that the utilities of internal nodes depend solely on the utilities of their children. In some situations, this assumption may be too restrictive. For example, while a laboratory may indeed derive its utility only from the utilities of its research groups, one could also argue that a laboratory may have its own objectives, such as expanding its office space independently of how the research groups use it. It would therefore be interesting to incorporate a personal component into the utilities of internal nodes. However, such an extension would significantly alter the model; in particular, the estimated utilities of internal nodes might no longer satisfy the MRF property.

Finally, another natural question is whether efficient algorithms can be designed for other fairness notions or under different assumptions on the leaves’ utilities. Indeed, \textbf{SMA} crucially relies on the existence of monolevel algorithms that jointly guarantee utilitarian optimality and fairness. However, such algorithms are not known for all popular fairness notions. For instance, while such an algorithm exists for EF1 \cite{Benabbou_et_al2021}, none is known for WEF1 to the best of our knowledge. As a consequence, \textbf{SMA} cannot compute a multilevel allocation that is both utilitarian-optimal and WEF1 w.r.t. $v$. Similarly, if the leaves’ utilities are no longer MRF, \textbf{SMA} is unlikely to remain applicable. Indeed, MRF utilities have a particularly convenient structure that enables algorithms with strong guarantees, whereas more classical utility classes generally do not allow such properties.

\bibliographystyle{plain}
\bibliography{acmart}

\appendix

\section{Glossary of notations} \label{appendix: glossary notations}

In Table~\ref{tab:glossary}, we propose a glossary of notations to help keep track of the different objects defined in the paper. 

\begin{table*}[h]
    \centering
    \begin{tabular}{ll}
    \toprule
    \textbf{Description} & \textbf{Notation} \\
    \midrule
    Set of items & $\items$ \\
    Set of agents & $\nodes$ \\
    Weight of agent $i \in \nodes$ & $w_i$ \\
    Multilevel tree rooted in $i$ & $\tree_i = (\nodes_i, \edges_i)$ \\
    Internal nodes of $\tree_i$ & $\internal(i)$ \\
    Leaves of internal node $i \in \internal(i)$ & $\leaves(i)$ \\
    Children of node $i \in \nodes$ & $\children(i)$ \\
    Ancestors of node $i \in \nodes$ & $\ancestors(i)$ \\
    Set of local allocations of $S \subseteq \items$ to $N \subseteq \nodes$ & $\allocations^S_N$ \\
    Set of multilevel allocations & $\mallocations$ \\
    Utility function of node $i \in \nodes$ & $v_i : \mallocations \rightarrow \mathbb{R}_{\geq 0}$ \\
    Utility function of leaf $x \in \leaves$ & $u_x : 2^\items \rightarrow \mathbb{R}_{\geq 0}$ \\
    Estimated utility of node $i \in \internal$ for SMA & $\hat{v}_i : 2^\items \rightarrow \mathbb{R}_{\geq 0}$ \\
    Fairness notion associated with node $i \in \internal(i)$ & $\Psi_i$ \\
    \bottomrule
    \end{tabular}
    \caption{Glossary of notation.}
    \label{tab:glossary}
\end{table*}

\section{The considered fairness notions w.r.t. $\hat{v}$} \label{appendix: fairness v-hat}

We define how the notions based on $v$, presented in Section~\ref{section: model}, can be adapted to $\hat{v}$. 

\begin{definition}[Lorenz-dominance w.r.t. $\hat{v}$] Given a bundle of items $S \subseteq \items$ and an internal node $i \in \internal$ whose children are $\{i_1, \ldots, i_p\}$, allocation $A \in \allocations^S_{\children(i)}$ is \emph{Lorenz-dominating w.r.t. $\hat{v}$} if for all $A' \in \allocations^S_{\children(i)}$, we have:
$$
\forall k \in \{1,\ldots, p\}, \sum_{t=1}^k \vv{\hat{s}(A)}_t \geq \sum_{t=1}^k \vv{\hat{s}(A')}_t
$$
where $\vv{\hat{s}(\cdot)}$ is the vector $\vv{\hat{v}_i(\cdot)}$ sorted in increasing order, and $\vv{\hat{s}(\cdot)}_t$ is the $t^{\text{th}}$ component of $\vv{\hat{s}(\cdot)}$ for any $t \in \{1, \ldots, p\}$ 
\end{definition}

\begin{definition}[weighted leximin w.r.t. $\hat{v}$] Given a bundle $S \subseteq \items$ and an internal node $i \in \internal$, allocation $A \in \allocations^S_{\children(i)}$ is \emph{weighted leximin w.r.t. $\hat{v}$} if 
$$
\nexists A' \in \allocations^S_{\children(i)}, \vv{\hat{e}(A')} \succeq_{lex} \vv{\hat{e}(A)}
$$
where $\vv{\hat{e}(\cdot)}$ is the vector $(\frac{\vv{\hat{v}_i(\cdot)}_1}{w_{i_1}}, \ldots, \frac{\vv{\hat{v}_i(\cdot)}_p}{w_{i_p}})$ sorted in increasing order.
\end{definition}

\begin{definition}[maximum weighted Nash social welfare w.r.t. $\hat{v}$] Given a bundle of items $S \subseteq \items$ and an internal node $i \in \internal$, $A \in \allocations^S_{\children(i)}$ \emph{maximizes weighted Nash social welfare w.r.t. $\hat{v}$} if it minimizes the number of children with zero utility, and subject to that satisfies the following property:
$$
\forall A' \in \allocations^S_{\children(i)},
\prod_{j \in \children(i)} \hat{v}_{j}(A(j))^{w_{j}} \geq \prod_{j \in \children(i)} \hat{v}_{j}(A'(j))^{w_{j}}
$$
\end{definition}

\begin{definition}[maximum weighted p-means welfare w.r.t. $\hat{v}$] Given a bundle of items $S \subseteq \items$ and an internal node $i \in \internal$, $A \in \allocations^S_{\children(i)}$ \emph{maximizes weighted $p$-means welfare w.r.t. $\hat{v}$} for some $p \leq 1$ if it minimizes the number of children with zero utility, and subject to that satisfies the following property:
$$
\forall A' \in \allocations^S_{\children(i)},
(\sum_{j \in \children(i)} w_{j} \hat{v}_{j}(A(j))^p)^{\frac{1}{p}} \geq (\sum_{j  \in \children} w_{j} \hat{v}_{j}(A'(j))^p )^{\frac{1}{p}}
$$

\end{definition}

\section{Missing proof from Section~\ref{section: sma}} \label{appendix: missing proofs section 3}

We now prove Proposition~\ref{proposition:v hat  = and sup v}. Let $\pi^*$ be the multilevel allocation constructed by {\bf SMA}. 

\vhatsupv*

\begin{proof}
Recall that $v_i(\pi) = \sum_{j \in \children(i)} v_j(\pi) = \sum_{x \in \leaves(i)}  u_x(\pi(x))$ for any allocation $\pi \in \mallocations$. From $\pi|_{\leaves(i)} \in \allocations^{\pi(i)}_{\leaves(i)}$ and $\hat{v}_i(\pi(i)) = \max_{A \in \allocations^{\pi(i)}_{\leaves(i)}} \sum_{x \in \leaves(i)} u_x(A(x))$, we can derive $\hat{v}_i(\pi(i)) \geq v_i(\pi)$ which establishes the second part of the proposition.

For any node $i \in \nodes$, let $h(i)$ be the height of $i$, i.e. the number of arrows in a longest path starting at node $i$. Let us prove by induction that property  $p(h)$ : "$\hat{v}_i(\pi^*(i)) = v_i(\pi^*)$ for all $i\in \internal$ s.t. $h(i)=h$" holds for all $h \in \{1,\ldots, h(1)\}$. We begin by proving the base case : let $i \in \internal$ be an internal node such that $h(i)=1$. Note that $\children(i) = \leaves(i)$ by definition and that $\pi^*\vert_{\children(i)} \in \arg \max_{A \in \allocations^{\pi^*(i)}_{\children(i)}} \sum_{j \in \children(i)} \hat{v}_j(A(j))$ since $\pi^*\vert_{\children(i)}$ is utilitarian optimal w.r.t. $\hat{v}$ in $\mathcal{A}^{\pi^*(i)}_{\children(i)}$ (see lines 4 and 5 in Algorithm 1). Hence $\pi^*\vert_{\leaves(i)} \in \arg \max_{A \in \allocations^{\pi^*(i)}_{\leaves(i)}} \sum_{x \in \leaves(i)} u_x(A(x))$ and so
\begin{align*}
    & \hat{v}_i(\pi^*(i)) = \max_{A \in \allocations^{\pi^*(i)}_{\leaves(i)}}\sum_{x \in \leaves(i)} u_x(A(x)) \mbox{ by definition of }\hat{v}_i \\
    & = \sum_{x \in \leaves(i)} u_x(\pi^*(x)) \mbox{ since } \pi^*\vert_{\leaves(i)} \in \arg \max_{A \in \allocations^{\pi^*(i)}_{\leaves(i)}} \sum_{x \in \leaves(i)} u_x(A(x))\\
    & = \sum_{x \in \leaves(i)} v_x(\pi^*) \mbox{ by definition of } v_x\\
    & = v_i(\pi^*) \mbox{ by definition of } v_i
\end{align*}
Thus we have $\hat{v}_i(\pi^*(i)) =  v_i(\pi^*)$ for any internal nodes $i$ with $h(i)=1$, which
shows that $p(1)$ is true. To prove the inductive step, we first need to prove the following lemma:

\begin{lemma} \label{lemma: usw children = usw leaves}
For any internal node $i \in \internal$, we have
 $$\max_{A \in \allocations^{\pi^*(i)}_{\children(i)}} \sum_{j \in \children(i)} \hat{v}_j(A(j)) = \max_{A \in \allocations^{\pi^*(i)}_{\leaves(i)}} \sum_{x \in \leaves(i)} u_x(A(x))$$
\end{lemma}

\begin{proof}
In this proof, we consider two types of local allocations: allocations to leaves in $\tree_i$ and  allocations to the children of $i$. 
\begin{itemize}
    \item For any local allocation  $A \in \allocations^{\pi^*(i)}_{\leaves(i)}$ to the leaves in $\leaves(i)$, we denote by $W_{\leaves(i)}(A)$ its utilitarian social welfare w.r.t. $u$, i.e. $W_{\leaves(i)}(A) = \sum_{x \in \leaves(i)} u_x(A(x))$. Let $A^* \in \allocations^{\pi^*(i)}_{\leaves(i)}$ denote a local allocation to the leaves in $\leaves(i)$ with maximum welfare, i.e. $W_{\leaves(i)}(A^*) = \max_{A \in \allocations^{\pi^*(i)}_{\leaves(i)}} \sum_{x \in \leaves(i)} u_x(A(x))$.
    \item For any local allocation $B \in \allocations^{\pi^*(i)}_{\children(i)}$ to the children in $\children(i)$, we denote by $\hat{W}_{\children(i)}(B)$ its utilitarian social welfare w.r.t. $\hat{v}$, defined by $\hat{W}_{\children(i)}(B) = \sum_{j \in \children(i)} \hat{v}_j(B(j))$. Let $B^* \in \allocations^{\pi^*(i)}_{\children(i)}$ be a local allocation to the children in $\children(i)$ with maximum welfare, i.e. $\hat{W}_{\children(i)}(B^*) = \max_{B \in \allocations^{\pi^*(i)}_{\children(i)}} \sum_{j \in \children(i)} \hat{v}_j(B(j))$.
\end{itemize}
Our aim is to prove that $W_{\leaves(i)}(A^*) = \hat{W}_{\children(i)}(B^*)$. First, let us show that $W_{\leaves(i)}(A^*) \le \hat{W}_{\children(i)}(B^*)$. Let $D \in \allocations^{\pi^*(i)}_{\children(i)}$ be the local allocation to the children defined by $D(j) = \cup_{x \in \leaves(j)} A^*(x)$ for all $j \in \children(i)$. 
With a slight abuse of notation, we write $A^*\vert_{\leaves(j)}$ to denote the allocation to leaves in $\leaves(j)$ such that $A^*\vert_{\leaves(j)}(x) = A^*(x)$ for all $x \in \leaves(j)$. Note that $A^*\vert_{\leaves(j)}$ belongs to $ \allocations^{D(j)}_{\leaves(j)}$ for any $j \in \children(i)$ and since 
$\hat{v}_j(D(j)) = \max_{A \in \allocations^{D(j)}_{\leaves(j)}} \sum_{x \in \leaves(j)} u_x(A(x))$, then we necessarily have  $\hat{v}_j(D(j)) \geq \sum_{x \in \leaves(j)} u_x(A^*(x))$. By summing over all $j\in \children(i)$, we obtain $\hat{W}_{\children(i)}(D) = \sum_{j \in \children(i)} \hat{v}_j(D(j)) \geq \sum_{j \in \children(i)} \sum_{x \in \leaves(j)} u_x(A^*(x)) = W_{\leaves(i)}(A^*)$ by linearity. Since $B^*$ is an allocation to the children that maximizes $\hat{W}_{\children(i)}$ by definition, then we have $\hat{W}_{\children(i)}(B^*) \ge \hat{W}_{\children(i)}(D) \ge W_{\leaves(i)}(A^*)$.

Now let us prove that $W_{\leaves(i)}(A^*) \geq \hat{W}_{\children(i)}(B^*)$ holds. For any $j\in \children(i)$, let $C_j$ denote a local allocation to the leaves in $\leaves(j)$ such that $\hat{v}_j(B^*(j)) = \max_{A \in \allocations^{B^*(j)}_{\leaves(j)}} \sum_{x \in \leaves(j)} u_x(A(x)) = \sum_{x \in \leaves(j)} u_x(C_j(x))$. Let $C$ denote the local allocation to the leaves in $\leaves(i)$ composed of all local allocations $C_j$, $j\in \children(i)$. More precisely, for any $x\in \leaves(i)$, we have $C(x) = C_{j}(x)$ where $j \in \children(i)$ is such that $x \in \leaves(j)$. Then, we obtain $W_{\leaves(i)}(C) = \sum_{j \in \children(i)} \sum_{x \in \leaves(j)} u_x(C_j(x)) = \sum_{j \in \children(i)} \hat{v}_j(B^*(j))$ by definition of  $C_j$, $j\in \children(i)$. Since the latter sum is equal to $\hat{W}_{\children(i)}(B^*)$ by definition, then we obtain $W_{\leaves(i)}(C) = \hat{W}_{\children(i)}(B^*)$. Finally, since $A^*$ is a local allocation to the leaves in $\leaves(i)$ that maximizes $W_{\leaves(i)}$ by definition, then we obtain $W_{\leaves(i)}(A^*) \ge W_{\leaves(i)}(C)=\hat{W}_{\children(i)}(B^*)$ which concludes the proof.
\end{proof}

Now let us assume that $p(k)$ holds for all $k \in \{1,\ldots,h\}$ for a given $1 \le h < h(1)$, and let us prove that $p(h+1)$ holds. Let $i \in \internal$ be any internal node such that $h(i)=h+1$. Note that $\pi^*\vert_{\children(i)} \in \arg \max_{A \in \allocations^{\pi^*(i)}_{\children(i)}} \sum_{j \in \children(i)} \hat{v}_j(A(j))$ since $\pi^*\vert_{\children(i)}$ is utilitarian optimal w.r.t. $\hat{v}$ in $\allocations^{\pi^*(i)}_{\children(i)}$ (see lines 4 and 5 in Algorithm 1). Hence we have the following equalities: 
\begin{align*}
   \hat{v}_i(\pi^*(i))  &=\max_{A \in \allocations^{\pi^*(i)}_{\leaves(i)}} \sum_{x \in \leaves(i)} u_x(A(x)) \mbox{ by definition of }\hat{v}_i\\
     & = \max_{A \in \allocations^{\pi^*(i)}_{\children(i)}} \sum_{j \in \children(i)} \hat{v}(A(j)) \mbox{ by Lemma~\ref{lemma: usw children = usw leaves} }\\
    & = \sum_{j \in \children(i)} \hat{v}_j(\pi^*(j)) \mbox{ by construction of }\pi^*|_{\children(i)} \\
    & = \sum_{j \in \children(i)} v_j(\pi^*) \mbox{ by induction since } h(j) \le h\\
    & = v_i(\pi^*) \mbox{ by definition of }v_i
\end{align*}
Thus $p(h+1)$ is true which concludes the proof. 
\end{proof}

\section{On Envy-Freeness Up to One Good} \label{appendix: on EF1}

In the paper, we focus on a family of fairness notions that we grouped under the notation $\fair$. However, as mentioned at the end of Section~3, \textbf{SMA} can with other fairness notions, provided that (i) the notion admits a suitable adaptation to the multilevel setting and (ii) there exists a local monolevel algorithm ensuring this notion together with utilitarian optimality. One classical notion for which we can extend \textbf{SMA} is envy-freeness up to one good (EF1), using the algorithm of \cite{Benabbou_et_al2021} together with a multilevel adaptation of EF1. 
Envy-based notions indeed require careful adaptation to the multilevel setting: the utility functions of internal nodes cannot evaluate the bundle of a sibling, as they can only evaluate multilevel allocations. Here we propose an optimistic multilevel adaptation of EF1. It assumes that every internal node $i \in \internal$ is equipped with an estimated utility function that, given a set of items $S \subseteq \items$, computes a utilitarian-optimal allocation of $S$ to its leaves $\leaves(i)$ and returns the utilitarian welfare of this allocation. This optimistic estimated utility function coincides exactly with the estimated utility function $\hat{v}_i$ introduced in Section~3, and allows node $i$ to estimate the value of a bundle. We can then define multilevel envy-freeness up to one good as follows:

\begin{definition}[M{[opt]}-EF1 w.r.t. $v$]
    Given a multilevel allocation $\pi \in \Pi$ and an internal node $i \in \internal$, $\pi \vert_{\children(i)}$ is optimistic envy-free up to one good w.r.t. $v$ if for any $j, k \in \children(i)$, we have
    $$v_j(\pi) \geq \hat{v}_j(\pi(k) \setminus \{g\})$$ for some $g \in \pi(k)$. Multilevel allocation $\pi$ is then optimistic multilevel envy-free up to one good w.r.t. $v$ (denoted by M[opt]-EF1 w.r.t. $v$) if $\pi \vert_{\children(i)}$ is optimistic EF1 at every internal node $i \in \internal$.
\end{definition} 

Similarly to our other fairness notions, we then need to define optimistic EF1 w.r.t. $\hat{v}$. 

\begin{definition}[M{[opt]}-EF1 w.r.t. $\hat{v}$]
    Given a bundle $S \subseteq \items$ and an internal node $i \in \internal$, allocation $A \in \allocations^S_{\children(i)}$ is optimistic EF1 w.r.t. $\hat{v}$ if for any two children $j, k \in \children(i)$, we have 
    $$\hat{v}_j(A(j)) \geq \hat{v}_j(A(k) \setminus \{g\})$$ for some $g \in A(k)$. Multilevel allocation $\pi$ is then M[opt]-EF1 w.r.t. $\hat{v}$ if $\pi \vert_{\children(i)}$ is optimistic EF1 w.r.t. $\hat{v}$ at every internal node $i \in \internal$.
\end{definition}

The \textbf{SMA} algorithm consists of computing, from the root down to the leaves, at each node an allocation that is both utilitarian-optimal and M{[opt]}-EF1 w.r.t. $\hat{v}$ (using the algorithm in \cite{Benabbou_et_al2021}). We can then show it \textbf{SMA} returns a multilevel allocation that is both utilitarian-optimal and M[opt]-EF1 w.r.t. $v$ using Proposition~\ref{proposition:v hat  = and sup v}. 

\section{Gain functions for the fairness criteria} \label{appendix: gain functions}

Table~\ref{tab:summary gain functions} lists the gain functions associated with each fairness criterion, as established in \cite{Viswanathan_Zick2023b}. 

\begin{table*}[h!]
    \centering
    \caption{Summary of the gain functions.}
    \label{tab:summary gain functions}
    \begin{tabular}{|c|c|}
        \hline
        Fairness criterion $\Psi_i$ & Gain function $\phi_i$ \\
        \hline
        Lorenz dominance & $-v_i(\pi)$ \\
        \hline
        Weighed leximin & $\frac{-v_i(\pi)}{w_i}$; \\
        & Break ties using least index \\
        \hline
        Weighted Nash & $(1 + \frac{1}{v_i(\pi)})^{w_i}$ if $v_i(\pi) > 0$; \\
        & a large $M$ otherwise \\
        \hline
        Weighted $p$-Means & $\mathtt{sign}(p) \times w_i \times ((v_i(\pi)+1)^p - v_i(\pi)^p)$ if $v_i(\pi) > 0$ or $p > 0$; \\
        & $M w_i$ for a large $M$ otherwise \\
        \hline
    \end{tabular}
\end{table*}

The same authors also establish loss functions associated with some fairness notions. We give two of them for some internal node $i \in \internal$, one of its children $j \in \children(i)$, and a multilevel allocation $\pi \in \mallocations$: (1) for weighted leximin, $\theta_i(\pi, j) = (\frac{-v_j(\pi) + 1}{w_j}, -w_j)$, and (2) for min weighted p-means welfare for $p \geq 1$, $\theta_i(\pi, j) = w_i[(-v_j(\pi) + 1)^p - (-v_j(\pi))^p]$.

\section{Extending SMA to Chores} \label{appendix: chores}

In this section, we propose an extension of both algorithms to the chore setting.

\medskip 

\noindent \textbf{Utilities.} We now assume that the marginal utility of any item $g \in \items$ is either $-1$ or $0$. More precisely, the utility of each internal node is still defined as the utilitarian welfare of its children, while each leaf $x \in \leaves$ is endowed with a $\{-1,0\}$-SUB valuation function:

\begin{definition}[$\{-1, 0\}$-SUB function]
    A set function $u: 2^\items \rightarrow \mathbb{R}_{\leq 0}$ is a $\{-1, 0\}$-SUB function if it satisfies:
    \begin{enumerate}
        \item $u(\emptyset) = 0$,
        \item $\Delta^u(S, g) \in \{-1, 0\}$ for any $S \subseteq \items$ and any chore $g \in \items \setminus S$,
        \item $\Delta^u(S, g) \geq \Delta^u(T, g)$ for any $S \subseteq T \subseteq \items$ and any chore $g \in \items \setminus T$.
    \end{enumerate}
\end{definition}

\noindent \textbf{Fairness and loss functions.} As in \cite{Viswanathan_Zick2023b}, our algorithms for chores uses loss functions to guide the allocation towards a fair outcome. For any node $i\in \nodes \backslash\{1\}$, let $\theta_{\mathcal{P}(i)}(\pi, i)$ denote the {\em loss function} of agent $i$, capturing the marginal loss from receiving a chore under allocation $\pi$. Note that this function depends on the fairness notion of $\Psi_{\mathcal{P}(i)}$, and returns a $b$-dimensional vector computed from $v_i(\pi)$. Some loss functions for the aforementioned fairness criteria are provided in Appendix~\ref{appendix: gain functions}.

\medskip

\noindent \textbf{Algorithms.} Contrary to the setting with goods, where the goal was to compute a non-redundant multilevel utilitarian-optimal and $\Psi$-max allocation, the objective changes in the presence of chores: we aim to compute a multilevel allocation that is complete and $\Psi$-max w.r.t. $v$. We start by presenting an adaptation of \textbf{SMA} to chores, following the same scheme as the one used in \cite{Viswanathan_Zick2023b} to extend the General Yankee Swap to this setting.

First, we observe that, for any leaf $x \in \leaves$, a MRF valuation $\alpha_x: 2^\items \rightarrow \mathbb{R}_{\geq 0}$ can be constructed from its $\{-1,0\}$-SUB valuation as follows: 
$$\forall S \subseteq \items, \quad \alpha_x(S) = u_x(S) + |S|$$

Then, for any internal node $i \in \nodes$ and any multilevel allocation $\pi$, we define a useful alternative utility function $\beta_i: \Pi \rightarrow \mathbb{R}_{\geq 0}$ as follows:  
\[
\beta_i(\pi) =
\begin{cases}
\sum_{j \in \children(i)} \beta_j(\pi) & \text{if } i \in \internal, \\
\alpha_i(\pi(i)) & \text{otherwise}.
\end{cases}
\]

We can now describe the adapted version of \textbf{SMA}. First, we run \textbf{SMA}$(1, \pi, \beta, \Psi)$, with $\pi(1) = \items$, and $\pi(i) = \emptyset, \forall i \in \nodes \setminus \{1\}$, where $\beta = (\beta_i)_{i \in \nodes}$. This produces a multilevel allocation in which as many chores with marginal utility zero as possible are allocated. Then, we allocate each of the remaining unallocated chores sequentially by selecting nodes from the root to the leaves using the loss functions at every internal node. Once all chores are allocated, we claim that the resulting multilevel allocation is complete and $\Psi$-maximizing with respect to $v$. The corresponding pseudocodes are given in Algorithms~\ref{algo:find_leaf_chores} and~\ref{algo:chore SMA}.

\begin{algorithm}[]
\small
\caption{\small Select\_Leaf\_Chores}
\label{algo:find_leaf_chores}
\begin{algorithmic}[1]
    \State \textbf{Input}: $\tree$ - a multilevel tree ; $i \in \nodes$ - a node; $v$ -- valuations of all nodes; $\Psi$ -- fairness criteria of internal nodes
    \State \textbf{Output} : $x \in \leaves(i)$ - a leaf in $\tree_i$
    \If{$\children(i) = \emptyset$} 
    \State \Return $i$ \EndIf
    \If{$i = 1$} 
    \State $N \gets \arg \underset{j \in \children(i) \setminus \{0\}}{\min} \theta_i(\pi, j)$
    \Else 
    \State $N \gets \arg \underset{j \in \children(i)}{\min} \theta_i(\pi, j)$ \EndIf
    \State $j' \gets \max_{j \in N} j$
    \State \Return $\mathtt{Select\_Leaf\_Chores}(\tree, j', v, \Psi)$
\end{algorithmic}
\end{algorithm}

\begin{algorithm}[] \small
\caption{\small SMA for Chores}
\label{algo:chore SMA}
\begin{algorithmic}[1]
    \State \textbf{Input:} $i$ -- a node in $\nodes$; $\pi $ -- a multilevel allocation; $v$ -- valuations of all nodes; $\Psi$ -- fairness criteria of internal nodes
    \State \textbf{Output:} $\pi $ -- a multilevel allocation
    \State Run \textbf{SMA}$(1, \pi, \beta, \Psi)$
    \State $\mathcal{RI} \gets \pi(1) \setminus \cup_{i \in \children(1)} \pi(i)$ \Comment{Set of remaining items}
    \While{$\mathcal{RI} \neq \emptyset$}
        \State $i \gets \texttt{Select\_Leaf\_Chores}(\tree, 1, v, \Psi)$
        \State Let $g$ be any unallocated chore, i.e. $g \in \mathcal{RI}$
        \While{$i \neq 1$}
            \State $\pi(i) \gets \pi(i) \cup \{g\}$
            \State $i \gets \parent(i)$
        \EndWhile
        \State $\mathcal{RI} \gets \mathcal{RI} \setminus \{g\}$
    \EndWhile
\end{algorithmic}
\end{algorithm}

\begin{restatable}{theorem}{choresma}
    \textbf{SMA} for Chores computes a multilevel complete and $\Psi$-maximizing allocation in polynomial time.
\end{restatable}

\begin{proof}
    We first note that the proof is simply the multilevel adaptation of the proof of \cite{Viswanathan_Zick2023b}, which shows that their extension of the General Yankee Swap can handle chores. Denote $\pi$ the multilevel allocation returned by Algorithm~\ref{algo:chore SMA}, and call $\pi'$ a complete $\Psi$-maximizing multilevel allocation. We will show that at any internal node $i \in \internal$, we have $\pi \vert_{\children(i)} \succeq^v_{\Psi_i} \pi' \vert_{\children(i)}$, and hence ultimately that $\pi$ is multilevel $\Psi$-maximizing. 

    We first prove this for the root, i.e., for $i=1$  (the proof is identical for lower internal nodes). We know that the first step of Algorithm~\ref{algo:chore SMA} (running SMA with $\beta$) returns a multilevel utilitarian-optimal allocation w.r.t. the $\beta$ functions. We first show that it implies multilevel utilitarian-optimality w.r.t. the true valuations $v$.

    First, notice that $\pi$ is the allocation returned by SMA, in which additional items were allocated. However, those unallocated items were the chores with strictly negative marginal utility (under $v$), i.e. chores that had null marginal utility under $\beta$. Therefore, the allocation of those additional items with null marginal utility w.r.t. the $\beta$ functions does not change the value of the allocation. Hence, we know $\pi$ is multilevel utilitarian-optimal w.r.t $\beta$. We have: 
    \begin{align*}
        \beta_i(\pi) \geq \beta_i(\pi') 
        \Leftrightarrow &\sum_{j \in \children(i)} \beta_j(\pi) \geq \sum_{j \in \children(i)} \beta_j(\pi') \\
        \Leftrightarrow & \sum_{x \in \leaves(i)} \alpha_x(\pi(x)) \geq \sum_{x \in \leaves(i)} \alpha_x(\pi'(x)) \\
        \Leftrightarrow & \sum_{x \in \leaves(i)} u_x(\pi(x)) + |\pi(x)| \geq \sum_{x \in \leaves(i)} u_x(\pi'(x)) + |\pi'(x)| \\
        \Leftrightarrow & \ |\pi(i)| + \sum_{x \in \leaves(i)} u_x(\pi(x)) \geq |\pi'(i)| + \sum_{x \in \leaves(i)} u_x(\pi'(x)) \\
        \Leftrightarrow & \ v_i(\pi) \geq v_i(\pi')
    \end{align*}

    Now, we prove that we have $v_j(\pi) \geq v_j(\pi')$ for all $j \in \children(i)$, and therefore by Pareto-Dominance of our fairness notions, we have $\pi \vert_{\children(i)} 
    \succeq^v_{\Psi_i} \pi' \vert_{\children(i)}$. If $v_j(\pi) \geq v_j(\pi')$ for all $j \in \children(i)$, the claim follows. Let us assume for contradiction that there exists $j \in \children(i)$ such that $v_j(\pi) < v_j(\pi')$, i.e. there exists a child of $i$ who received more chores in $\pi$ than in $\pi'$. If there are multiple such node, pick the one with the least $\theta_i(\pi, j)$. Since $v_i(\pi) = \sum_{j \in \children(i)} v_j(\pi) \geq v_i(\pi') = \sum_{j \in \children(i)} v_j(\pi')$, it must be that there exists a node $k \in \children(i)$ such that $v_k(\pi) > v_k(\pi')$. Since utilities are non-positive, it must be that $0 \geq v_k(\pi) > v_k(\pi')$, and since $v_k(\pi') < 0$, there must exist a chore $g \in \pi'(k)$ such that one of its leaves $x \in \leaves(k)$ has $\Delta^{u_x}(\pi'(x) \setminus \{g\}, g) = -1$. 

    Let $\pi''$ be the same allocation than $\pi'$, but move chore $g$ from $\pi'(k)$, and the bundle of any of the descendants of $k$ who got chore $g$, to the bundle of a leaf $y \in \leaves(j)$ (and by extension to the bundle of any of the ancestors in $\ancestors(y)$. If $\Delta^{u_y}(\pi'(y), g) = 0$, then $\pi'' \vert_{\children(i)} \succ^v_{\Psi_i} \pi' \vert_{\children(i)}$, a contradiction with $\pi'$ being $\Psi$-max. So, assume $\Delta^{u_y}(\pi'(y), g) = -1$. 

    We show next that $\theta_i(\pi', j) \leq \theta_i(\theta'', k)$, which implies $\pi'' \vert_{\children(i)} \succ^v_{\children(i)} \pi' \vert_{\children(i)}$, by consistency, which would yield a contradiction. Let $\pi_1$ be the allocation maintained by Algorithm~\ref{algo:chore SMA} at the start of the iteration where $j$ received its final chore. Note that such iteration must exist since $v_j(\pi) < v_j(\pi') \geq 0$. Since $j$ was chosen at this iteration, it must be that $$\theta_i(\pi_1, j) \leq \theta_i(\theta_1, k). \text{ If equality holds, then $j < k$.}$$

    Furthermore, we have $$\theta_i(\pi', j) \leq \theta_i(\pi_1, j) \leq \theta_i(\pi_1, k) \leq \theta_i(\pi'', k)$$ where the first inequality comes from $v_j(\pi_1) = v_j(\pi) + 1 \leq v_j(\pi')$, and the last inequality stems from $v_k(\pi_1) \geq v_k(\pi) \geq v_k(\pi') + 1 = v_k(\pi'')$. If any inequality is strict, then $\theta_i(\pi', j) < \theta_i(\pi'', k)$, and therefore we have $\pi'' \vert_{\children(i)} \succ^v_{\Psi_i} \pi' \vert_{\children(i)}$, a contradiction. Otherwise, we have $\theta_i(\pi_1, j) = \theta_i(\pi_1, k)$, and therefore $j < k$. Hence, $g$ should have been allocated to $g$ rather than $k$, and therefore $\pi'' \vert_{\children(i)} \succ^v_{\Psi_i} \pi' \vert_{\children(i)}$, again a contradiction. Hence, it must be that $v_j(\pi) \geq v_j(\pi')$  for all $j \in \children(i)$. This argument applies to any $i \in \internal$. Hence, for any $i \in \internal$, we proved $\pi \vert_{\children(i)}$ is $\Psi_i$-maximizing, and therefore $\pi$ is multilevel $\Psi$-maximizing. 
\end{proof}

\vspace{0.2cm}
Note that the \textbf{MGYS} can also be extended following the same technique: it suffices to run \textbf{MGYS} instead of \textbf{SMA} in Line~3 of Algorithm~\ref{algo:chore SMA}. Since \textbf{MGYS} is a heuristic, it does not provide formal guarantees beyond (trivial) completeness, but it will run faster than the \textbf{SMA} for Chores. Note that Lines 4-13 will allocate exactly the same chores to the same nodes both in the \textbf{SMA} for Chores and the \textbf{MGYS} for Chores. Indeed, Line 3 (whether it is running \textbf{SMA} or \textbf{MGYS}) will allocate as many 0-valued chores as possible, and the remaining non-allocated chores are all -1-valued. Hence, the allocation computed by either \textbf{SMA} or \textbf{MGYS} will yield 0 utility for all nodes, and therefore each chore-version algorithm will select the very same leaf and item at every iteration to complete the allocation (since they are guided by the same loss functions).

\end{document}